\documentclass[acmsmall,nonacm]{acmart}

\usepackage{cleveref}
\usepackage{bcprules}
\usepackage{xspace}
\usepackage{listings}
\usepackage{mathtools}
\usepackage{rulelinks}
\usepackage{enumitem}
\usepackage{subcaption}
\definecolor{blue-violet}{rgb}{0.54, 0.17, 0.89}
\definecolor{depmap}{HTML}{00007B}
\definecolor{dark-cyan}{HTML}{135579}
\definecolor{magenta}{HTML}{a8264f}
\colorlet{aout}{purple}
\newcommand*{\algout}[1]{{\color{aout}{#1}}}
\newcommand{\newcontent}[1]{{\setlength{\fboxsep}{1.7pt}%
  \colorbox{pink}{\ensuremath{#1}}}
}

\makeatletter
\newcommand{\note}[1]{{\if@ACM@timestamp\color{red}[#1]\else\fi}}
\makeatother
\newcommand{\bfparagraph}[1]{\paragraph{\textbf{\textsf{#1}}}}
\newcommand{\eg}{\emph{e.g.}}
\newcommand{\ie}{\emph{i.e.}}

\newcommand{\qbot}{\ensuremath{\varnothing}}

\newcommand{\qfresh}[1][-2]{\text{\larger[#1]\ensuremath{\vardiamondsuit}}}
\newcommand{\qhole}[1][]{\ensuremath{\square{#1}}}
\newcommand{\qclean}[1]{\ensuremath{{#1}\!\setminus\!\qfresh}}
\newcommand{\overlap}{\ensuremath{\mathbin{\phantom{l}\mathclap{\qfresh[-5]}\mathclap{\cap}\phantom{l}}}}
\newcommand{\qsat}[1]{\ensuremath{#1\mathord{*}}}
\newcommand{\WF}[1]{\ensuremath{#1\ \mathsf{ok}}}
\newcommand{\reaches}{\ensuremath{\mathrel{\leadsto}}}
\newcommand{\qexp}[1][]{\ensuremath{\mathrel{\Uparrow\ifstrempty{#1}{}{_{#1}}}}}
\newcommand{\qunif}{\ensuremath{\mathrel{\subseteq?}}}

\newcommand{\flt}{\ensuremath{\varphi}}
\newcommand{\ts}{\ensuremath{\vdash}}
\newcommand{\mts}{\ensuremath{\models}}
\newcommand{\ots}[1]{\ensuremath{\dashv\algout{#1}}}
\newcommand{\cx}[2][]{\ensuremath{\ifstrempty{#1}{#2}{{#2}^{\,#1}}}}
\newcommand{\pcx}[2][]{\cx[#1]{\ifstrempty{#1}{#2}{\left(#2\right)}}}

\newcommand{\G}[1][]{\cx[#1]{\Gamma}}
\newcommand{\dom}[1]{\ensuremath{\mathsf{dom}\,#1}}

\newcommand{\subs}{\ensuremath{\subseteq}}
\newcommand{\subq}[1][:]{\ensuremath{<#1}}
\newcommand{\subt}[1][:]{\ensuremath{%
  \ifstrequal{#1}{:}{\leqslant#1}{\leqslant_{#1}}}}
\newcommand{\subapp}[2][:]{\ensuremath{\prec#1^{\,#2}}}
\newcommand{\synth}{\ensuremath{\mathrel{\Rightarrow}}}
\newcommand{\chek}{\ensuremath{\mathrel{\Leftarrow}}}
\newcommand{\growth}[1]{\ensuremath{\nearrow^{\,\algout{#1}}}}

\newcommand{\maybelang}{\ensuremath{\mathsf{G}_{<:}^{\qfresh}}\xspace}
\newcommand{\polylang}{\ensuremath{\mathsf{F}_{<:}^{\qfresh}}\xspace}
\newcommand{\algolang}{\ensuremath{\mathsf{G}_{\Leftrightarrow}^{\qfresh}}\xspace}

\newcommand{\judgement}[2]{{\textsf{\textbf{#1}}} \hfill #2}
\newcommand{\BOX}[1]{\fbox{\ensuremath{\strut #1}}}

\newcommand{\ty}[2][]{\ensuremath{\ifstrempty{#1}{#2}{{#2}^{\,#1}}}}
\newcommand{\pty}[2][]{\ty[#1]{\ifstrempty{#1}{#2}{(#2)}}}
\newcommand{\ctxvar}[2]{\ensuremath{#1\!:\!#2}}
\newcommand{\tqvar}[1]{\ty[\!\lowercase{#1}]{\uppercase{#1}}}
\newcommand{\ctxtyp}[2]{\ensuremath{\tqvar{#1}\!\!<:\!#2}}

\newcommand{\TBase}[1][]{\ty[#1]{\mathsf{B}}}
\newcommand{\TRef}[3][]{\pty[#1]{\mathsf{Ref}~#2\ifstrempty{#3}{}{..\;\!#3}}}
\newcommand{\TTop}[1][]{\ty[#1]{\mathsf{Top}}}
\newcommand{\TFun}[5][]{\pty[#1]{#2(\ctxvar{#3}{#4})\!\to#5}}
\newcommand{\TAll}[5][]{\pty[#1]{\forall#2[\ctxtyp{#3}{#4}].~#5}}
\newcommand{\TSelf}[3][]{\pty[#1]{\mu#2.\,#3}}
\newcommand{\TSRef}[4][]{\TSelf[#1]{#2}{\TRef{#3}{#4}}}
\newcommand{\tref}[1]{\ensuremath{\mathsf{ref}~#1}}
\newcommand{\tget}[1]{\ensuremath{\mathsf{!}~#1}}
\newcommand{\tput}[2]{\ensuremath{#1~\mathsf{:=}~#2}}
\newcommand{\tlam}[4][]{\ensuremath{\lambda#2(%
  \ifstrempty{#1}{#3}{\ctxvar{#3}{#1}}).~#4}}
\newcommand{\ttlam}[4][]{\ensuremath{\Lambda#2[%
  \ifstrempty{#1}{\tqvar{#3}}{\ctxtyp{#3}{#1}}].~#4}}
\newcommand{\tapp}[2]{\ensuremath{#1~#2}}
\newcommand{\ttapp}[2]{\ensuremath{#1~[#2]}}
\newcommand{\tast}[2]{(\ctxvar{#1}{#2})}

\lstdefinelanguage{PolyRT}{
  morekeywords={%
    abstract,fun,%
    case,char,class,%
    def,else,extends,final,finally,%
    if,import,implicit,%
    match,module,%
    new,null,ref,%
    object,override,%
    package,private,protected,public,%
    public,return,super,%
    this,trait,%
    val,var,%
    while,%
    yield,%
    let,end,%
    alloc,%
    move,%
    map,%
  },%
  mathescape,%
  sensitive,%
  morecomment=*[l]//,%
  morecomment=*[s]{/*}{*/},%
  morestring=[b]",%
  morestring=[b]',%
  literate=%
    {~*}{{\qfresh}}1%
    {~0}{{\qbot}}1%
    {<=}{{\subt}}1%
    {==>}{{\synth}}3%
    {<==}{{\chek}}3%
    {|-}{{\ts}}1%
    {...}{{$\cdots$}}3%
    {=eq}{{$\equiv$}}1%
    {|->}{{$\mapsto$}}1%
    {->}{{$\to$}}1%
    {<-}{{$\leftarrow$}}1%
    {\\A}{{$\forall$}}1%
    {<@}{{\subapp[]{}}}1%
    {\\mu}{{$\mu$}}1%
    {\\ne}{{$\nearrow$}}2%
    {<<}{{$\ll$}}1%
    {\\OVLP}{{$\overlap$}}1%
    {\\E}{{$\exists$}}1%
    {<->}{{$\iff$}}3%
    {!<}{{$\color{red}\nless$}}1%
}[keywords,comments,strings]%

\lstdefinestyle{PolyRT}{%
  basicstyle=\footnotesize\ttfamily,%
  showstringspaces=false,%
  keywordstyle={\color{dark-cyan}\bf\ttfamily},%
  stringstyle=\color{darkgray},%
  commentstyle=\color{dark-cyan},%
  escapebegin=\color{dark-cyan},%
  moredelim=*[is][\color{dark-cyan}]{`}{'},%
  moredelim=[is][\color{red}]{!!}{!!},
  moredelim=*[il][\color{aout}]??,
  moredelim=[is][\color{dark-cyan}\textsuperscript]{^}{^},
  moredelim=[is][\color{red}\textsuperscript]{^!}{^},
  moredelim=[is][\color{aout}\textsuperscript]{^?}{^},
  columns=fixed,%
  fontadjust=true,%
  basewidth=0.5em,%
  aboveskip=2pt,
  belowskip=2pt,
}%
\newenvironment{lstequation}%
  {\noindent\begin{minipage}{\linewidth}%
   \footnotesize\[\tt}%
  {\]\vspace{-6pt}\end{minipage}}
 
\setcopyright{cc}
\setcctype{by}
\acmDOI{10.1145/3808335}
\acmYear{2026}
\acmJournal{PACMPL}
\acmVolume{10}
\acmNumber{PLDI}
\acmArticle{257}
\acmMonth{6}
\acmSubmissionID{pldi26main-p703-p}
\received{2025-11-14}
\received[accepted]{2026-04-03}

\makeatletter%
\begin{document}

\title[Escape with Your Self]{Escape with Your Self: Sound and Expressive Bidirectional Typing with Avoidance for Reachability Types}

\author{Songlin Jia}
\orcid{0009-0008-2526-0438}
\affiliation{%
  \institution{Purdue University}
  \city{West Lafayette}
  \country{USA}
}
\email{jia137@purdue.edu}

\author{Guannan Wei}
\orcid{0000-0002-3150-2033}
\affiliation{%
  \institution{Tufts University}
  \city{Medford}
  \country{USA}
}
\email{guannan.wei@tufts.edu}

\author{Siyuan He}
\orcid{0009-0002-7130-5592}
\affiliation{%
  \institution{Purdue University}
  \city{West Lafayette}
  \country{USA}
}
\email{he662@purdue.edu}

\author{Yuyan Bao}
\orcid{0000-0002-3832-3134}
\affiliation{%
  \institution{Augusta University}
  \city{Augusta}
  \country{USA}
}
\email{yubao@augusta.edu}

\author{Tiark Rompf}
\orcid{0000-0002-2068-3238}
\affiliation{%
  \institution{Purdue University}
  \city{West Lafayette}
  \country{USA}
}
\email{tiark@purdue.edu}

\begin{abstract}

Reasoning about programs in the presence of mutation and aliasing is 
notoriously difficult. 
Rust has popularized lifetime-based ownership tracking in systems programming, 
but its ``shared XOR mutable'' model is fundamentally at odds with 
higher-level functional programming.
Reachability types offer an alternative: they enable safe sharing and 
escape of mutable data by tracking which resources each expression's result can reach.

To track internal reachability within complex object graphs,
reachability types adopt \emph{self-references} that let
components refer to enclosing resources from inside,
just like \texttt{this} pointers in OO languages.
While natural for \emph{declaratively} typing escaping data,
self-references complicate subtyping and furthermore type inference:
variance restricts where self-references may appear, yet useful type
conversions must allow them to vary in controlled ways, which in turn imposes
constraints on inference.
As an undesirable result, prior works require programmers to insert term-level
coercions for even just \emph{avoidance}---avoiding ill-scoped names in types.

With all prior works being declarative,
we investigate \emph{algorithmic} reachability types in this work.
We introduce a refined subtyping relation that permits more flexible usages
of self-references.
We further develop a sound and decidable bidirectional typing algorithm,
implemented and verified in Lean.
The algorithm automatically avoids ill-scoped names in types,
and infers qualifiers via a lightweight unification mechanism.
As a step towards practical reachability programming,
we show that the system is capable of tracking diverse reachability patterns
without explicit coercions in complex Church-encoded datatypes.
 \end{abstract}



\maketitle
\lstMakeShortInline[]@

\section{Introduction} \label{sec:intro}

Mutable state with possible aliasing enables expressive programming
patterns, but is also non-trivial to reason about, leading to memory safety
violations and resource leaks.
For this reason, there has been a surge of interest in language
designs that regulate aliasing or mutability through a type system
\cite{DBLP:series/lncs/ClarkeOSW13, DBLP:conf/oopsla/ClarkePN98,
DBLP:conf/oopsla/TschantzE05, DBLP:conf/oopsla/ZibinPLAE10}.
Rust \cite{DBLP:conf/sigada/MatsakisK14}, the most notable example, has shown
that lifetime tracking based on ownership types is an
eminently practical way of ensuring memory safety in a low-level imperative
system language.
To realize a similar flavor of lifetime reasoning in higher-order languages,
however, adapting Rust's approach would be restrictive: its
``\emph{shared XOR mutable}''
model---permitting a resource to be either shared or mutable, but
not both---prohibits many functional programming idioms that require
capturing and sharing mutable values.

Reachability types
\cite{DBLP:journals/pacmpl/BaoWBJHR21,DBLP:journals/pacmpl/WeiBJBR24,DBLP:journals/pacmpl/0001HJBR25,DBLP:journals/pacmpl/BaoJ0BR25}
are a recent proposal to bring the benefits of
lifetime reasoning and principled aliasing to higher-level languages.
The key idea is to augment the type system with a qualifier that tracks
which resources an expression's result can reach.
Reachability types impose restrictions on sharing and separation
at a fine-grained, per-expression level, rather than as global ownership
invariants.
This permits mutable data to be shared or to escape its defining scope
while retaining static safety guarantees, even in the presence of higher-order
functions and polymorphic types.
Below, we exemplify the programming patterns this approach enables,
including those beyond the \emph{shared XOR mutable} model.

\paragraph{Controlling Lifetimes}

Reachability types naturally support resource-management patterns
where a handle must remain scoped to its defining context.
Consider a @withFile@ combinator that opens a file, passes the handle
to a callback, and closes it upon return.
In the type of @withFile@, the freshness marker~$\qfresh$ signifies that
the file handle is new to the callback and not reachable in other ways:
\begin{lstlisting}
// withFile: \A[A^a^].(path: String) -> (body: File^~*^ -> A^a^) -> A^a^
withFile[Unit]("log.txt")(fun f => write(f)("hello") )
\end{lstlisting}
Soundness requires that the file handle not escape its defining scope
\cite{DBLP:conf/ecoop/XhebrajB0R22,DBLP:conf/oopsla/OsvaldEWAR16,DBLP:journals/corr/abs-2207-03402},
and the type of the callback @body@ enforces this:
the argument @f@ is a fresh file handle and
lies outside the domain of the return type~%
@A^a^@, so no instantiation of @A^a^@ for @withFile@ can allow a call leaking @f@:
\begin{lstlisting}
let f = withFile[File^~*^]("log.txt")(fun f => f )   // !!Error:!! f not allowed in result type
\end{lstlisting}

This guarantee is analogous to Rust's lifetime-based scoping
using higher-ranked trait bounds (HRTB).
The two systems enforce the same invariant here,
but diverge when sharing is introduced.

\paragraph{Sharing and Separation}

A useful pattern is to define several writers over the same file
handle.  For example, consider loggers at different severity levels:
\begin{lstlisting}
withFile[Unit]("log.txt")(fun f =>
  let warn  = fun (msg: String) => write(f)("WARN: " ++ msg)   // warn:  (String => Unit)^f^
  let error = fun (msg: String) => write(f)("ERROR: " ++ msg)  // error: (String => Unit)^f^
  warn("low memory"); error("disk full") )
\end{lstlisting}
Reachability types track that both @warn@ and @error@ reach the handle @f@,
allowing them to safely coexist.
This tracking enables static enforcement of separation:
in the type of a parallel combinator @par@~%
\cite{DBLP:journals/pacmpl/WeiBJBR24,DBLP:journals/pacmpl/BaoJ0BR25},
freshness on both parameters requires that each not be
reachable from the other. The type system accepts the call when
qualifiers are disjoint and rejects it when they overlap, for example:
\begin{lstlisting}
// par: (f: (Unit -> Unit)^~*^) -> (g: (Unit -> Unit)^~*^) -> Unit, requires f and g separate
par(fun _ => warn("starting") )(fun _ => error("failed") )     // !!Error:!! both reach f
\end{lstlisting}

In Rust, a mutable resource cannot be borrowed twice,
so the pattern above requires falling back to dynamic
reference-counting, \eg, using @Rc<RefCell<T>>@ or @Arc<Mutex<T>>@, erasing the
aliasing relationship from the types of @warn@ and @error@.
Guarding against conflicting borrows then requires runtime mechanisms.

\paragraph{Into Stateful Objects}

Sharing and separation reasoning extends to data structures that capture
resources.
In a functional setting, objects can be naturally modeled as products of
closures with shared state, making it
essential to track their common, hidden resources.
The example below encodes a @Counter@ class
\cite{DBLP:journals/pacmpl/WeiBJBR24,DBLP:journals/pacmpl/XuBO24} with a
hidden state @c@ and two mutation methods.
Reachability types track that both methods access~@c@ without referring
to @c@ outside its defining scope.  Instead, a \emph{self-reference} $p$ is
introduced on the result pair using the $\mu$-notation, replacing @c@ in the
qualifiers of the two closures:
\begin{lstlisting}
let Counter = fun (n: Int) =>                //: Int -> \mup.Pair[(Unit -> Unit)^p^,(Unit -> Unit)^p^]^~*^
  let c = ref n
  Pair(fun () => c := !c + 1, fun () => c := !c - 1)  //: Pair[(Unit -> Unit)^c^,(Unit -> Unit)^c^]^c^
\end{lstlisting}
When the returned pair is bound to the variable @ctr@, its projected
components are then understood to reach @ctr@ as well, thus preserving the
sharing invariant:
\begin{lstlisting}
let ctr = Counter(0)                         // ctr: Pair[(Unit -> Unit)^ctr^,(Unit -> Unit)^ctr^]^ctr^
let incr = fst(ctr); let decr = snd(ctr)
par(fun _ => incr() )(fun _ => decr() )      // !!Error:!! both methods reach ctr
\end{lstlisting}

Rust would require dynamic reference-counting for such a case
(cf. Figure 1b, \citet{DBLP:journals/pacmpl/BaoWBJHR21}).

\bfparagraph{From Self-References to Algorithms}

While reachability types have shown great promise for expressiveness,
practical type checking and inference algorithms have not been addressed:
prior works focused on declarative formulations and proofs
of type soundness, with no algorithmic counterpart.
The main source of difficulty is \emph{self-references}---qualifiers that
let components refer to their enclosing resources (like @this@ pointers),
adopted from path-dependent types in DOT \cite{DBLP:conf/oopsla/RompfA16,
DBLP:conf/birthday/AminGORS16}.
As seen in the @Counter@ example, self-references are essential for
tracking escaping data, but they complicate both the subtyping needed
for \emph{avoidance}\footnote{
Avoidance converts types to remove variables about to go out of scope.
In the \lstinline@Counter@ example, this means replacing \lstinline@c@ with the
self-reference $p$ when the pair is returned.
} and the inference of qualifiers (see \Cref{sec2:depfun}).

Prior work
\cite{DBLP:journals/pacmpl/WeiBJBR24,DBLP:journals/pacmpl/0001HJBR25,DBLP:journals/pacmpl/BaoJ0BR25}
treats types and qualifiers separately in subtyping, leaving reachability opaque
and self-references inert; avoidance thus requires term-level coercions.
Challenging their assumption that \emph{the locations a value may reach are fixed}
motivates a refined semantic model in which they depend on the type assigned to
it, a combined type-and-qualifier subtyping that enables avoidance conversions
involving self-references directly in subtyping, and algorithmic procedures for
automated qualifier inference and avoidance built atop this new foundation.

\bfparagraph{Contributions and Organization}

We address the open challenges of type checking/inference with avoidance
for reachability types
\cite{DBLP:journals/pacmpl/WeiBJBR24,DBLP:journals/pacmpl/0001HJBR25,DBLP:journals/pacmpl/BaoJ0BR25}
by (1) proposing the $\maybelang$-calculus
with combined type-and-qualifier subtyping, and (2) developing a sound and
decidable bidirectional typing algorithm \algolang
to the specification of \maybelang,
both fully mechanized in Lean.
Together, these systems enable end-to-end support for expressive
high-level constructs with only modest annotation overhead.

\begin{itemize}[leftmargin=1.75em]
  \item We review the core concept of reachability types, and
    identify key challenges in algorithm design on avoidance and
    qualifier inference (\Cref{sec:motiv}), motivating our design of \maybelang
    and \algolang.

  \item We introduce the $\maybelang$-calculus (\Cref{sec:lang}), featuring
    type-and-qualifier subtyping and qualifier holes in algorithmic contexts.
    These mechanisms enable flexible conversions involving
    self-references and allow type assignment under partially specified function
    qualifiers.

  \item We develop the bidirectional typing algorithm \algolang (\Cref{sec:algo}),
    with mechanized proofs of soundness and decidability relative to \maybelang.
    The algorithm infers qualifiers for expressions and avoids
    ill-scoped qualifiers in typing while tracking escaping data in higher-order settings.

  \item We evaluate \algolang on programs using higher-order
    functions and data structures (\Cref{sec:eval}), demonstrating
    expressiveness for resource and lifetime reasoning, with moderate
    annotation and performance overhead.
    We further discuss alternative designs and possible extensions.
\end{itemize}
We discuss related work in \Cref{sec:related} and conclude the paper
with \Cref{sec:conc}.
In our artifact \cite{artifact}, we provide
the mechanization of \maybelang and \algolang,
and examples for empirical evaluation.%
\footnote{Also available at
\url{https://github.com/TiarkRompf/reachability/tree/main/checking/lean_v2}.} %
\section{Motivation} \label{sec:motiv}

Reachability types
\cite{DBLP:journals/pacmpl/BaoWBJHR21,DBLP:journals/pacmpl/WeiBJBR24,DBLP:journals/pacmpl/0001HJBR25}
concern the use of resources in impure functional languages.
They track resources by \emph{reachability qualifiers} and
enable functions to constrain their arguments by sharing and separation.
While prior works \emph{declaratively} characterize the valid type
assignments of terms,
they do not spell out \emph{algorithmic} steps for checking such
assignments,
nor do they support inferring types and qualifiers for ergonomic programming.
In this section, we present an informal overview of our reachability
type system \algolang with bidirectional typing and avoidance support.

\subsection{Elements of Reachability Qualifiers} \label{sec2:elements}

While exact resource identities---often memory locations---are unavailable
at compile time, reachability types \cite{DBLP:journals/pacmpl/WeiBJBR24,DBLP:journals/pacmpl/BaoJ0BR25}
approximate them statically using
(1)~\emph{variables} for resources with externally visible names,
(2)~\emph{freshness} for unnamed resources, \eg, new allocations,
and (3)~\emph{self-references} for components of unnamed resources 
to refer to their enclosing resources or themselves.

\vspace{-1pt}
\paragraph{Variables}

When the new allocation @ref 42@ is assigned to the variable @b@,
we infer (@`==>'@) all later uses of~@b@ to be @??Ref[Int]^?b^@,%
\footnote{As a convention, we use {\color{aout}purple} to emphasize
results inferred by the algorithm.}
where the qualifier @??b@ signifies that
values resulting from the expression @b@
reach no more resources than the established variable~@b@:

\begin{lstlisting}
let b = ref 42; b            // [b: Ref[Int]^~*^] |- b     ==> ??Ref[Int]^?b^
\end{lstlisting}

\looseness=-1
Qualifiers of functions include their free variables, like @b@ in the
example below, reflecting the fact that at runtime,
the assignments of such free variables are recorded (thus reached) in the closure:

\begin{lstlisting}
fun () => b := !b + 1        // [b: Ref[Int]^~*^] |- <fun> ==> ??(`Unit' -> Unit)^?b^
\end{lstlisting}

\vspace{-1pt}
\paragraph{Freshness}

Expressions marked fresh~@~*@ represent the resources
\emph{unreachable} from existing variables,
so that operating on fresh values causes no interference.
New allocations by @ref@ are always fresh.
Below, both @ref 42@ (later reached by @b@)
and @ref 43@ yield fresh references, guaranteed to be separate:

\begin{lstlisting}
let b = ref 42; ref 43       // [b: Ref[Int]^~*^] |- ref 43 ==> ??Ref[Int]^?~*^     `<- separate from b'
\end{lstlisting}

Freshness markers in argument qualifiers represent \emph{contextual} freshness:
they are unreachable from variables reached in the function,
but may overlap with variables not observed by the function.
Below, @accumulate@ takes a fresh argument @x@.
Thus, parameters to @accumulate@ should be separate from its captured values, 
specifically @acc@, but may still reach other variables:

\begin{lstlisting}
let b = ref 42; let acc = ref 0
let accumulate =
  fun (x: Ref[Int]^~*^) =>      // [..., acc: Ref[Int]^~*^] |- <fun> ==> ??(`(x: Ref[Int]^~*^)' -> Unit)^?acc^
    acc := !acc + !x
accumulate(b)                // Okay:  'b' not observed by 'accumulate'
accumulate(acc)              // !!Error:!! 'acc' not separate from 'accumulate'
\end{lstlisting}

\vspace{-1pt}
\paragraph{Self-References}

Just like @this@ pointers in object-oriented languages allowing
fields to be reached by methods,
self-references allow resources to be reached from inside,
which is crucial to representing escaping data structures.
Below, we model an object @obj@ with a mutable state and a @Pair@ of
@getter@/@setter@ methods.
While the name @b@ is not visible outside, we still need to track that
the methods reach shared resources.
As the @Pair@ becomes the new logical owner of the reference @b@,
we use its self-reference @??p@ introduced by the $\mu$-notation
to qualify the pair components:

\begin{lstlisting}
let obj =                    // [..., obj: ...^~*^] |- obj ==> ??\mup.Pair[(Unit -> Int)^?p^, (Int -> Unit)^?p^]^?obj^
  { let b = ref 42
    Pair(fun () => !b, fun (n: Int) => b := n) }
let getter = fst(obj)        // [..., obj: ...^~*^] |- fst(obj) ==> ??(Unit -> Int)^?obj^  // p mapped to obj
let setter = snd(obj)        // [..., obj: ...^~*^] |- snd(obj) ==> ??(Int -> Unit)^?obj^
\end{lstlisting}

\subsection{Challenges for Algorithmic Reachability Types}
\label{sec2:depfun}

Reachability types refer to term variables in types, making them dependent.
However, without term evaluation in types,
they are faced with the \emph{avoidance problem}
known in bounded existential types~\cite{DBLP:journals/tcs/GhelliP98},
DOT \cite{DBLP:conf/oopsla/RompfA16,DBLP:conf/birthday/AminGORS16},
and module systems
\cite{DBLP:journals/pacmpl/BlaudeauRR25,DBLP:conf/popl/DreyerCH03,lillibridgeTranslucentSumsFoundation,DBLP:journals/jfp/RossbergRD14,DBLP:journals/jfp/Leroy00}:
they need mechanisms to remove names that are about to go out of scope,
\eg, avoiding @b@ in the @obj@ example above.

To deal with avoidance, reachability types allow
(1)~substituting variables with qualifiers
and (2)~converting types using self-references.
Nevertheless, both mechanisms are nontrivial for inference algorithms.
We analyze these mechanisms and their challenges in the rest of this section.

\vspace{-2pt}
\paragraph{Substituting Variables}

As a minimal example of dependent functions,
@identity@ returns its argument @x@ and captures no free variable.
Its result type reaches the bound variable~@??x@:

\begin{lstlisting}
let identity = fun (x: Ref[Int]^~*^) => x  // [...   ] |- <fun> ==> ??(`(x: Ref[Int]^~*^)' -> Ref[Int]^?x^)^?~0^
\end{lstlisting}
After function application, @x@ is not defined in the scope and thus can no
longer occur. To preserve reachability tracking,
we can substitute @x@ in the type with parameter qualifiers,
@b@ or @~*@, respectively:

\begin{lstlisting}
let b = ref 42; identity(b)  // [..., b: Ref[Int]^~*^] |- <app> ==> ??Ref[Int]^?b^      // Ref[Int]^x^ [b/x]
identity(ref 42)             // [...              ] |- <app> ==> ??Ref[Int]^?~*^      // Ref[Int]^x^ [~*/x]
\end{lstlisting}
Such \emph{dependent application} achieves lightweight reachability
polymorphism~\cite{DBLP:journals/pacmpl/WeiBJBR24} ergonomically via function
parameters~\cite{Rytz:175546,DBLP:journals/toplas/BoruchGruszeckiOLLB23},
in addition to explicit polymorphism via separate quantification.

Variables bound by @let@ can be substituted similarly.
Below, we replace @b@ by @~*@ in the result type:

\begin{lstlisting}
let b = ref 42               // [...              ] |- <let> ==> ??Ref[Int]^?~*^      // Ref[Int]^b^ [~*/b]
b                            // [..., b: Ref[Int]^~*^] |- b     ==> ??Ref[Int]^?b^
\end{lstlisting}

\vspace{-2pt}
\paragraph{Avoidance Conversion by Self-References}

Substitution with freshness is restricted within types.
Take \emph{the escaping closure} below for example:
the function expression returns the captured reference~@b@.
When typing the overall let binding, the bound variable @b@ can no longer occur
in the resulting type, but substituting both~@b@'s with @~*@ would change the
type to mean returning new references:

\begin{lstlisting}
let b = ref 42               // [...              ] |- <let> !!cannot be!! (Unit -> Ref[Int]^?~*^)^?~*^
fun () => b                  // [..., b: Ref[Int]^~*^] |- <fun> ==>       ??(Unit -> Ref[Int]^?b^)^?b^
\end{lstlisting}
To prevent such an unintended change, we note that substitution with @~*@ cannot
take place inside types: it is the substitution in the function
return qualifier that causes the change.

To avoid the internal occurrence of @b@ without resorting to substitution,
we approximate it using self-references.
We convert~(@`<<'@) the type of the @let@-body to use the self-references @f@,%
\footnote{In the function type \lstinline@f(Unit) -> ...@,
we define \lstinline@f@ as its self-reference.
We omit $\mu$-notations for functions.}
so that substituting @b@ is only required for the top-level qualifier:

\begin{lstlisting}
let b = ref 42               // [...              ] |- <let> ==>       (f(Unit) -> Ref[Int]^f^)^?~*^
fun () => b                  // [..., b: Ref[Int]^~*^] |- <fun> ==> ... << ??(f(Unit) -> Ref[Int]^?f^)^?b^
\end{lstlisting}

To summarize, substitutions on bound variables are subject to the restriction
that they must be \emph{either {non-fresh}, or {non-deep}}.
This restriction is already seen in prior work
\cite{DBLP:journals/pacmpl/WeiBJBR24,DBLP:journals/pacmpl/0001HJBR25}
from their declarative systems.  Algorithmically,
when a deep occurrence needs substitution involving freshness, it has to be
first removed by avoidance conversions using self-references.

\bfparagraph{Algorithmic Challenges}

Both mechanisms are nontrivial to implement.
Substitution requires \emph{inferring precise parameter qualifiers}
for precise substitution results.
Avoidance conversion requires devising a \emph{type conversion scheme} that 
removes undesired variables and is sound with respect to a subtyping relation.

\subsubsection{Inferring Qualifiers for Parameters}

In bidirectional typing systems \cite{DBLP:journals/csur/DunfieldK21} without
reachability, when applying a function like @identity@, the parameter is
checked (@`<=='@) against the argument type:

\begin{lstlisting}
let b = ref 42; identity(b)  // [..., identity: Ref[Int] -> Ref[Int], ...] |- b <== Ref[Int]
\end{lstlisting}

For reachability types,
we need to adopt a hybrid checking/inference mode (@`<=='^?==>^@) similar to
that seen in refinement types \cite{DBLP:conf/pldi/PolikarpovaKS16},
checking the type but inferring the qualifiers:

\begin{lstlisting}
let b = ref 42; identity(b)  //[...,identity:(x:Ref[Int]^~*^)-> Ref[Int]^x^,...] |- b <== Ref[Int]^==>^^?b^
\end{lstlisting}
Although inferring @??b@ in @identity(b)@ is straightforward,
it gets more complicated when self-references are involved in argument types,
which introduces constraints from typing the parameter.

\paragraph{Constraints from Self-References}

Below, @callGet@ calls its argument @get@ to retrieve what it
captures:

\begin{lstlisting}
let callGet =  fun (get: ( `g'(Unit)  ->  Ref[Int]^g^ )^~*^) => get() //: ??Ref[Int]^?get^
\end{lstlisting}
With the result of @get@ qualified by its self-reference @g@ and the result of 
@callGet@ qualified by @get@, both results may reach the same resource as @get@.
Such constraints among reachability can be alternatively understood using
an explicit qualifier quantification, illustrated below as @callGet2@:

\begin{lstlisting}
let callGet2 = fun [`q <: ~*'](get: (Unit -> Ref[Int]^q^)^q^) => get() //: ??Ref[Int]^?q^
\end{lstlisting}

Invoking @callGet2@ requires instantiating the qualifier variable
@`q'@ to satisfy the typing:

\begin{lstlisting}
let b = ref 42; callGet2`[b]'(fun _ => b)  //[...,b:Ref[Int]^~*^] |- <fun> <== (Unit -> Ref[Int]^q^)^q^ [b/q]
\end{lstlisting}
This is analogous to instantiating type variables for type polymorphism,
whose inference is known to be nontrivial
\cite{DBLP:journals/toplas/PierceT00,DBLP:conf/popl/OderskyZZ01,DBLP:conf/icfp/DunfieldK13,DBLP:journals/pacmpl/ZhaoOS19,DBLP:conf/ecoop/ZhaoO22,DBLP:journals/pacmpl/CuiJO23}.
While \algolang does not deal with type instantiations,
it infers parameter qualifiers for invoking @callGet@,
achieving similar expressiveness ergonomically:

\begin{lstlisting}
let b = ref 42; callGet(fun _ => b)      //[...,b:Ref[Int]^~*^] |- <fun> <== (g(Unit) -> Ref[Int]^g^)^==>^^?b^
\end{lstlisting}

We elaborate on our approach to qualifier inference in \Cref{sec2:inferqual}.

\subsubsection{Self-References Conversions for Avoidance}
\label{sec2:avoidintro}

As analyzed earlier, typing the \emph{escaping closure example} requires the
following type conversion to remove the variable @b@ from inside the type:

\begin{lstequation}
  [\cdots,\, b: Ref[Int]^{\qfresh}] \quad \ts \quad
  {\color{dark-cyan}(Unit \to Ref[Int]^{\color{red}b})^b}
  \quad \ll \quad
  {\color{dark-cyan}({\color{red}f}(Unit) \to Ref[Int]^{\color{red}f})^b}
\end{lstequation}
With @b@ included in the function qualifier, the self-reference @f@
then replaces~@b@ in the result qualifier.
For such conversions, we need to ensure that they are sound,
and that they can be applied generally.

\paragraph{Soundness}

Prior work
\cite{DBLP:journals/pacmpl/WeiBJBR24,DBLP:journals/pacmpl/0001HJBR25}
could not justify the conversion resulting in a supertype,
largely due to the fact that their subtyping relations ($\subq$) are
designed without top-level qualifiers:

\begin{lstequation}
  [\cdots,\, b: Ref[Int]^{\qfresh}] \quad \ts \quad
  {\color{dark-cyan}Unit \to Ref[Int]^{b}}
  \quad {\color{red}\not{\color{black}<:}} \quad
  {\color{dark-cyan}{f}(Unit) \to Ref[Int]^{f}}
\end{lstequation}
Unable to see that \emph{the function qualifier includes} @b@
in subtyping, prior works require $\eta$-expanding escaping closures
for \emph{retyping} them.
Manifested in programming with data structures like pairs,
this necessitates term-level coercions like @conv@ below,%
\footnote{Adapted from Section 8.1, \citet{DBLP:journals/pacmpl/WeiBJBR24}.
We elaborate further in \Cref{sec:vsprior}.}
making it unideal for both theory and practice:

\begin{lstlisting}
def conv[A^a^, B^b^](p: Pair[A^a^, B^b^]^{a,b,~*}^): \mup.OPair[A^p^, B^p^]^{a,b}^ =
  OPair(fst(p), snd(p))  // reconstruct Pair to use self-references in types (OPair)
\end{lstlisting}
To justify seamless type conversions, we thus need to improve the
notion of subtyping.
More than just reinstating missing qualifiers,
we need to deal with the variance of self-references.

\paragraph{Generalizing for Variance}

The \emph{escaping closure example} requires avoidance in function result qualifiers,
which is the only position where prior works
\cite{DBLP:journals/pacmpl/WeiBJBR24,DBLP:journals/pacmpl/0001HJBR25}
allow a self-reference.
In contrast, variables about to go out of scope may occur in places with
different variances and depths:

\begin{lstlisting}
let b = ref 42                    // fresh 'b' about to go out of scope
fun () => b                       // covariant:      Unit -> Ref[Int]^b^
fun (x: Ref[Int]^b^) => !x          // contravariant:  Ref[Int]^b^ -> Int
ref b                             // invariant:      Ref[Ref[Int]^b^]
fun (f: Ref[Int]^b^ -> Int) => f(b)  // deep covariant: (Ref[Int]^b^ -> Int) -> Int
\end{lstlisting}
Generally avoiding them requires a systematic scheme
to use self-references.
We discuss our avoidance algorithm in \Cref{sec2:avoid}, and
its soundness foundation---our new subtyping---in
\Cref{sec2:subtyping}.

\subsection{Self-Reference Conversions for Avoidance}
\label{sec2:avoid}

At a high-level, our algorithmic avoidance involves
replacing variables in covariant positions by self-references
and simply removing the contravariant ones.
We detail our solution as follows.

\paragraph{Covariant Occurrences}

For the base case from \Cref{sec2:avoidintro},
we replace @`b'@ in the covariant result qualifier with the
self-reference @??f@, and
we add @??b@ in the function qualifier if it is not yet included:

\begin{lstlisting}
let fn = { let b = ref 42    // fn    ==> ??(f(Unit) -> Ref[Int]^?f^)^?fn^
           fun () => b    }  // <fun> ==> (Unit -> Ref[Int]^b^)^b^  <<  (f(Unit) -> Ref[Int]^?f^)^b,^^?b^
\end{lstlisting}
When the escaping closure is later named @fn@ and applied, the self-reference
@f@ becomes another bound variable needing substitution. We replace it with the
qualifier of @fn@, \ie, @fn@ itself:

\begin{lstlisting}
fn()                         // <app> ==> ??Ref[Int]^?fn^                      // Ref[Int]^f^ [fn/f]
\end{lstlisting}

\paragraph{Contravariant Occurrences}

When the unwanted variable occurs in contravariant positions, we
remove it, analogous to replacing ill-scoped type
variables with the bottom type in algorithmic System $F_{<:}$
\cite{DBLP:journals/toplas/PierceT00}.
Illustrated below, we change the argument type from
@`Ref[Int]'^b^@ to @`Ref[Int]'^?~0^@:

\begin{lstlisting}
let fn = { let b = ref 42
           fun (x: Ref[Int]^b^) => !x }  // <fun> ==> (Ref[Int]^b^ -> Int)^b^  <<  (Ref[Int]^?~0^ -> Int)^b^
\end{lstlisting}

Such a conversion renders @fn@ non-callable, but it is necessary in this
specific example:
the function is originally defined to \emph{receive no more references than}
@b@, and there is no qualifier other than @b@ itself that can keep this
invariant.
In practice, a dummy function type like that of @fn@ should be a signal for
an overly conservative type annotation somewhere.

\paragraph{Invariant Occurrences}

Aligned with recent development
\cite{DBLP:journals/pacmpl/0001HJBR25,DBLP:journals/pacmpl/GaoP25},
we adopt the notion of \emph{dual-component} references, with
a contravariant \emph{put} type and a covariant \emph{get} type;
types with a single invariant referent are then seen as shorthands for two
components being the same.
To avoid @b@ in the example below, we consider the two components
separately: we remove @b@ in the \emph{put} qualifier,
replace @b@ with the self-reference @??h@ in the \emph{get} qualifier,
and add @b@ to the top-level reference qualifier:

\begin{lstlisting}
let r = { let b = ref 42     // r     ==> ??\muh.Ref[Ref[Int]^?~0^..Ref[Int]^?h^]^?r^
          ref b          }   // ref b ==> Ref[Ref[Int]^b^]^~*^  <<  \muh.Ref[Ref[Int]^?~0^..Ref[Int]^?h^]^~*,^^?b^
\end{lstlisting}
To read the reference @r@ resulting from escaping, similar to when calling functions,
we need to replace @h@ with the actual qualifier @r@, signifying that the
extracted result is internal to @r@:

\begin{lstlisting}
!r                           // !r    ==> ??Ref[Int]^?r^                          // Ref[Int]^h^ [r/h]
\end{lstlisting}
By necessity, @r@ is made \emph{read-only} with the put qualifier @??~0@
to accept \emph{no more assignment than} @b@.

\paragraph{Deep Occurrences}

Deep, covariant uses of the unwanted variable are replaced with
the self-reference of the outermost function/reference.
Exemplified below, we replace @b@ inside the argument type with @??f@.
To apply this escaped @fn@,
@f@ in the argument type should first be replaced with @fn@:

\begin{lstlisting}
let fn = { let b = ref 42   // <fun> ==> ((Ref[Int]^b^ -> Int) -> Int)^b^ << (f(Ref[Int]^?f^ -> Int) -> Int)^b,^^?b^
           fun (get: Ref[Int]^b^ -> Int) => get(b) }
fn(fun b => !b)             // <fun> <== (Ref[Int]^?fn^ -> Int)^ ==>^^?fn^     // (Ref[Int]^f^ -> Int) [fn/f]
\end{lstlisting}

\subsection{Type-and-Qualifier Subtyping for Avoidance}
\label{sec2:subtyping}

To justify avoidance conversions,
we propose a combined form of type-and-qualifier subtyping,
with self-references enabled in covariant positions
for both expressiveness and soundness.

\paragraph{Combined Subtyping}

In our declarative specification \maybelang,
we use the symbol $\subt$ for the extended subtyping on qualified types.
The base covariant conversion case can be formalized as:

\begin{lstequation}
  [\cdots,\, b: Ref[Int]^{\qfresh}] \quad \ts \quad
  {\color{dark-cyan}(Unit \to Ref[Int]^{b})^b}
  \quad \subt \quad
  {\color{dark-cyan}({f}(Unit) \to Ref[Int]^{\color{red}f})^b}
\end{lstequation}

With both sides agreeing on the outermost qualifier @b@, this fact enables
proving the subtyping relation, \emph{packing} variable names into 
self-references.
\maybelang can also justify \emph{unpacking} that returns self-references into
variables, making post-avoidance types compatible with operations before:

\begin{lstequation}
  [\cdots,\, fn: (f(Unit) \to Ref[Int]^f)^{\qfresh}] \quad \ts \quad
  {\color{dark-cyan}({f}(Unit) \to Ref[Int]^{f})^{fn}}
  \quad \subt \quad
  {\color{dark-cyan}(Unit \to Ref[Int]^{\color{red}fn})^{fn}}
\end{lstequation}

\paragraph{Supporting Growing Qualifiers}

Qualifiers in subtyping may not fully agree, for example:

\begin{lstequation}
  [\cdots,\, a: Ref[Int]^{\qfresh},\, b: Ref[Int]^{\qfresh}] \quad \ts \quad
  {\color{dark-cyan}(Unit \to Ref[Int]^{b})^{a}}
  \quad \subt \quad
  {\color{dark-cyan}({f}(Unit) \to Ref[Int]^{f})^{a,b}}
\end{lstequation}
In \maybelang, this is still a valid subtyping relation, but requires a
transitivity step changing the qualifier
of the subtype (left) from @`a'@ to @`a,b'@,
via the type @`(Unit-> Ref[Int]^b^)^a,b^'@;
it is then upon the typing algorithm to find such necessary intermediate steps
to justify subtyping relations (see \Cref{sec2:inferqual}).

Moreover, with qualifiers growing in subtyping relations,
self-references in the two sides may carry differing interpretations.
Illustrated below, dictated by function qualifiers,
the self-reference @f@ in the subtype (left) can reach
at most @`a'@, but in the supertype (right) it may additionally reach~@`b'@:

\begin{lstequation}
  [\cdots,\, a: Ref[Int]^{\qfresh},\, b: Ref[Int]^{\qfresh}] \quad \ts \quad
  {\color{dark-cyan}(f(Unit) \to Ref[Int]^{f})^{a}}
  \quad \subt \quad
  {\color{dark-cyan}(f(Unit) \to Ref[Int]^{f})^{a,b}}
\end{lstequation}
Thus, in subtyping, self-references are \emph{covariant}.
In \maybelang, we simply require that self-references occur only in
covariant positions in types.
This aligns with our algorithmic avoidance scheme, which never involves
self-references contravariantly.

\subsection{Inferring Function Qualifiers} \label{sec2:inferqual}

Inference of qualifiers is most challenging for functions.
A function's qualifier is determined primarily by
(1) \emph{observations} collected from the function body,
and (2) constraints induced by subtyping relations that involve self-references.
Collecting observations is necessary for new function definitions, whereas
resolving constraints is required when we use existing functions in different
types.
We detail the two sources as follows.

\paragraph{Collecting Observations}

When evaluating expressions, resources may be transiently involved without
necessarily being reached in the resulting value.
Such \emph{observations} naturally include free variables of expressions.
In the example below, the expression @!b@ observes @??b@,
as marked beside the context on the right:

\begin{lstlisting}
let b = ref 42; !b     // [..., b: Ref[Int]^~*^]^?b^ |- !b ==> ??Int^?~0^
\end{lstlisting}
Expressions may also observe resources indirectly without referring to free
variables. In the example below, the nested reference @c@ allows the expression
to observe not only @??c@, but also @??b@:

\begin{lstlisting}
let c = ref b; !(!c)   // [..., c: Ref[Ref[Int]^b^]^~*^]^?b,c^ |- !(!c) ==> ??Int^?~0^
\end{lstlisting}

In \algolang, while typing expressions, we additionally collect their
observations.
Function qualifiers then include their body observations,
governing the resources they may access when invoked.

\begin{lstlisting}
let b = ref 42; let c = ref b
fun () => !(!c)        // [...]^?b,c^ |- <fun> ==> ??(Unit -> Int)^?b,c^
\end{lstlisting}

\paragraph{Constraints from Subtyping}

Given an existing function, subtype checking (\subt[]) is required when
we need to convert it into a different type, as exemplified below:

\begin{lstequation}
  [\cdots,\, a: Ref[Int]^{\qfresh},\, b: Ref[Int]^{\qfresh}] \quad \ts \quad
  {\color{dark-cyan}(Unit \to Ref[Int]^{b})^{a}}
  \quad \subt[] \quad
  {\color{dark-cyan}({f}(Unit) \to Ref[Int]^{f})^{\color{aout}???}}
\end{lstequation}
While the qualifier after conversion ($\scriptstyle\tt\color{aout}???$) needs
inference,
we know it should be at least @`a'@ and satisfy typing constraints from
the self-reference @f@ used in the type.

In \algolang, we take inspiration from the \emph{eager instantiation} approach
\cite{DBLP:conf/icfp/DunfieldK13} for higher-ranked polymorphic type inference.
We make the initial guess for $\scriptstyle\tt\color{aout}???$ to be
$\scriptstyle\tt\color{aout}a,\qhole$, where the \emph{qualifier hole} $\qhole$
stands for the existential reachability to be inferred,
analogous to unification variables for type inference
but occurring only in typing contexts.
In the typing context for checking function subtyping,
we use this initial guess to qualify the self-reference @f@:

\begin{lstequation}
  [\cdots,\, a: Ref[Int]^{\qfresh},\, b: Ref[Int]^{\qfresh},
    \, f: (\cdots)^{\color{aout}a,\qhole}] \quad \ts \quad
  {\color{dark-cyan}\phantom{Unit -> Ref} b }
  \quad \subq[] \quad
  {\color{dark-cyan}\phantom{Unit -> Ref} f}
\end{lstequation}

During the checking procedure,
the qualifier subsumption obligation @`b' < `f'@ poses a constraint to the
reachability of @f@.
To eagerly satisfy this constraint, we insert @??b@ into the hole
attached to @f@:

\begin{lstequation}
  [\cdots,\, a: Ref[Int]^{\qfresh},\, b: Ref[Int]^{\qfresh},
    \, f: (\cdots)^{\color{aout}a,b,\qhole}] \quad \ts \quad
  {\color{dark-cyan}\phantom{Unit -> Ref} b }
  \quad \subq[] \quad
  {\color{dark-cyan}\phantom{Unit -> Ref} f}
\end{lstequation}
When completing the function subtype checking,
the guess has been updated to
$\scriptstyle\tt\color{aout}a,b,\qhole$.
We then \emph{seal} the hole and use @??a,b@ as the final answer to the
qualifier inferred in subtype checking:

\begin{lstequation}
  [\cdots,\, a: Ref[Int]^{\qfresh},\, b: Ref[Int]^{\qfresh}] \quad \ts \quad
  {\color{dark-cyan}(Unit \to Ref[Int]^{b})^{a}}
  \quad \subt[] \quad
  {\color{dark-cyan}({f}(Unit) \to Ref[Int]^{f})^{\color{aout}a,b}}
\end{lstequation}

\subsection{Summary}

Self-references in reachability types are crucial for encoding data types,
whereas their presence in prior work hinders devising an avoidance conversion
scheme and complicates qualifier inference (\Cref{sec2:depfun}).  In this work,
we first present a refined declarative calculus \maybelang (\Cref{sec:lang})
that combines type-and-qualifier subtyping and enables self-references in all
covariant positions for expressiveness and soundness
(\Cref{sec2:subtyping,sec:subtyping}).  Building on this theoretical foundation, we
develop the algorithmic solution \algolang (\Cref{sec:algo}) that automatically
avoids ill-scoped variables by self-references in a polarity-guided fashion
(\Cref{sec2:avoid,sec4:avoidance}) and infers qualifiers for all expressions,
including new function definitions and their type conversions
(\Cref{sec2:inferqual,sec4:subtype,sec4:bidirectional}). %
\section{\maybelang: Declarative Typing Specification} \label{sec:lang}

We present the formal theory and metatheory of $\maybelang$, a refined variant
of the polymorphic reachability type system $\polylang$
\cite{DBLP:journals/pacmpl/WeiBJBR24}.
To serve as a specification for our algorithmic development, $\maybelang$
adopts shallow, dual-component reference types \cite{DBLP:journals/pacmpl/0001HJBR25},
introduces combined type-and-qualifier subtyping (\Cref{sec:subtyping}),
and uses algorithmic contexts with qualifier holes.
The type system and its soundness proofs are fully mechanized in Lean 4.
In the following sections, we detail the design of $\maybelang$, and briefly
discuss its metatheory in \Cref{sec3:meta}.

\subsection{Syntax Definitions and Well-Formedness}

\begin{figure}[t]
\footnotesize \makeatother
\judgement{Syntax}{\BOX{\maybelang}}\vspace{-7pt}
  \[\begin{array}{l@{\quad}l@{\quad}l@{\hspace{4em}}l}
    x,y,z     &\in & \mathsf{Var}                                                                            & \text{Variables} \\
    f,g,h     &\in & \mathsf{Var}                                                                            & \text{Self Variables} \\
    X,Y,Z     &\in & \mathsf{TVar}                                                                           & \text{Type Variables} \\
    \qhole{},\qhole{?} &\in & \newcontent{\mathsf{QHole}}                                                    & \text{(Optional) Qualifier Holes} \\
    t         &::= & c \mid x \mid \tref{t} \mid\ \tget{t} \mid \tput{t}{t} \mid \tapp{t}{t} \mid \ttapp{t}{Q} \mid \newcontent{\tast{t}{Q}} \mid & \text{Terms} \\
              &    & \tlam[Q^?]{f}{x}{t} \mid \ttlam[Q^?]{f}{X}{t}                                           & \text{(with optional domain types)} \\[1ex]
    p,q,r,s   &\in & \mathcal{P}_{\mathsf{fin}}(\mathsf{Var} \uplus \{ \qfresh \})                           & \text{Reachability Qualifiers} \\
    P,Q,R,S   &::= & \ty[q]{T}                                                                               & \text{Qualified Types} \\
    T,U,V,W   &::= & \TBase \mid \newcontent{\TSRef{h}{Q}{Q}} \mid \TFun{f}{x}{Q}{Q} \mid                    & \text{Types} \\
              &    & X \mid \TTop \mid \TAll{f}{X}{Q}{Q} \\[1ex]
    \Gamma    &::= & \varnothing \mid \G, \ctxvar{x}{Q} \mid \G, \ctxtyp{X}{Q} \mid \newcontent{\G, \ctxvar{f}{\ty[q,\qhole{?}]{T}}} & \text{Typing Environments} \\
    \flt      &\in & \mathcal{P}_{\mathsf{fin}}(\mathsf{Var})                                                & \text{Observations} \\
    \end{array}\] \\
\textbf{\textsf{Qualifier Shorthands}} \hfill
  $\begin{array}{l@{\hspace{7em}}}
    p,q := p \cup q \qquad
    x := \{x\} \qquad
    \qfresh :=\{\qfresh\} \qquad
    \qfresh{q} := \{\qfresh\}\cup q
  \end{array}$ \\[1ex]
\judgement{Substitution, Reachability and Overlap}{
    \BOX{q[p/x]}
  \ \BOX{{\color{gray}\G\vdash}\,x \reaches x}
  \ \BOX{{\color{gray}\G\vdash}\, \qsat{q}}
  \ \BOX{{\color{gray}\G\vdash}\,p \overlap q}}
  \[\begin{array}{l@{\quad}l@{\qquad}l@{\quad}l}
    \text{Qualifier Substitution} & q[p/x] := q{\setminus\{x\}}\cup\;p,\,\text{if~} x\in q; &
    q[p/x] := q, & \text{otherwise}. \\[1.2ex]
    \text{Reachability Relation} & {\color{gray}\G\vdash}\, x \reaches y \Leftrightarrow  x : T^{q,y} \in \Gamma &
    \text{Variable Saturation} & {\color{gray}\G\vdash}\, \qsat{x} := \left\{\, y \mid x \reaches^* y\, \right\} \\[1.1ex]
    \text{Qualifier Saturation} & {\color{gray}\G\vdash}\, \qsat{q} :=  \bigcup_{x\in q} \qsat{x} &
    \text{Qualifier Overlap} & {\color{gray}\G\vdash}\,p \overlap q  := \qfresh{(\qsat{p} \cap \qsat{q})}
  \end{array}\]
\vspace{-1em}
\caption{Syntax definitions of $\maybelang$ with qualifier shorthands
and qualifier operations. Sometimes the context $\G$ is implicit (in {\color{gray}gray}).
Adapted from \citet{DBLP:journals/pacmpl/WeiBJBR24}, \maybelang introduces 
qualifier holes in typing contexts and a type ascription term, and
adopts dual-component reference types \cite{DBLP:journals/pacmpl/0001HJBR25}.
We emphasize them as \newcontent{\text{shaded}}.
} \label{fig:decdefs}
\end{figure} 
\Cref{fig:decdefs} presents the syntax of $\maybelang$, which is based on
System~$F_{<:}$ with higher-order references.
In \Cref{fig:closed}, we provide the well-formedness
definitions for the syntax of $\maybelang$.

\paragraph{Terms.}
Terms include constants, variables, references and operations, function
abstractions and applications, and type abstractions and applications. For
clarity, we use distinct metavariables for ordinary variables (\ie, $x, y, z$),
functions (\ie, $f, g, h$), and type variables (\ie, $X, Y, Z$). Conventionally,
ordinary variables such as $x$ may represent functions, but not the other
way around. For function abstraction $\tlam[Q]{f}{x}{t}$, we read $f$ as the
self-reference of the function, and $x$ as the name of the argument;
the argument type $Q$ is optional.
Similarly for $\ttlam[{\ty[q]{T}}]{f}{X}{t}$, we
read $f$ as the self-reference, $X\!<: T$ as the type variable quantification 
and $x\!<: q$ as the qualifier quantification.
We also add type
ascription $(t: Q)$ to support the bidirectional typing algorithm (cf. \Cref{sec:algo}).

\paragraph{Qualifiers and Types.}
Qualified types ($Q$) consist of a type ($T$) paired with a qualifier ($q$).
Qualifiers are sets of variables that may include the freshness marker
$\qfresh{}$, denoting fresh values without names.
Types include the base type \TBase, references, functions, type variables, the
\TTop, and universal types.
The reference types are dual-component \cite{DBLP:journals/pacmpl/0001HJBR25}:
the first $Q$ describes the type for putting, and the second describes
getting; we simply write $\TRef{Q}{}$ if the two components are the same.

References, functions, and universal types also include self-references $h$ or
$f$ in their types, and we constrain the occurrences of such self-references.
As motivated in \Cref{sec2:subtyping} and illustrated by \rulename{cfun} in
\Cref{fig:closed}, duplicated here:

{ \footnotesize \setlength{\afterruleskip}{\smallskipamount}
  \infrule[\ruledef{cfun}{c-fun}]{
    \G,\ctxvar{f}{\TTop[\qfresh]}\ts \ty[p]{T} \qquad
    \G,\ctxvar{f}{\TTop[\qfresh]},\ctxvar{x}{\ty[p]{T}} \ts\ty[q]{U} \qquad
    \newcontent{f \notin^+ T} \qquad
    \newcontent{f \notin p} \qquad
    \newcontent{f \notin^- U}
  }{
    \G\ts \TFun{f}{x}{\ty[p]{T}}{\ty[q]{U}}
  }
}
The self-reference $f$ must
not appear in covariant positions of the domain $T$ (\ie, $f \notin^+ T$),
contravariant positions of the codomain $U$ (\ie, $f \notin^- U$), and
the domain qualifier $p$.
These polarity constraints are crucial to the soundness of our subtyping
extension.

\paragraph{Algorithmic Contexts.}
We use the metavariable $\G$ to denote typing contexts. In $\maybelang$, these
contexts are algorithmic \cite{DBLP:conf/icfp/DunfieldK13}: their entries are ordered,
allowing insertion and deletion in the back.

Self-reference entries in contexts may contain qualifier holes, \eg,
$\ctxvar{f}{\ty[q,\qhole]{T}} \in \G$.
For other entries, qualifiers must be fully specified, free of holes.
Definitions of well-formed contexts are given in \Cref{app:typing}, where
all entries must be closed under their preceding contexts.

\paragraph{Qualifier Shorthands and Operations.}
For presentation purposes, we define qualifier shorthands in \Cref{fig:decdefs},
allowing notation such as $x$ for singleton sets $\{x\}$ and $p, q$ for their
union $p \cup q$.
We also define qualifier substitution and overlap, both used in application
rules. Substitution on types is standard and thus omitted.
Qualifier overlap is defined as the intersection of transitive reachability closures,
and is used in the separation judgment \rulename{ffresh}.

\subsection{Typing} \label{sec:typing}

\begin{figure}[t]
\footnotesize\makeatother
\typicallabel{t} \setlength{\afterruleskip}{\smallskipamount}

\judgement{Term Typing (Selection)}{\BOX{\G[\flt] \ts t : Q}} \\[-1em]
\begin{tabular}{@{}p{.29\linewidth}@{}p{.32\linewidth}@{}p{.39\linewidth}@{}}
  \infrule[\ruledef{tcst}{t-cst}]{
    c \in \TBase
  }{
    \G[\flt] \ts c : \TBase[\qbot]
  }
  &
  \infrule[\ruledef{tvar}{t-var}]{
    \ctxvar{x}{\ty[q]{T}} \in \G \andalso
    x \in \flt
  }{
    \G[\flt] \ts x : \ty[x]{T}
  }
  &
  \infrule[\ruledef{tsub}{t-sub}]{
    \G[\flt] \ts t : P \phantom{mn}
    \G\ts P \subt \ty[q]{T} \phantom{mn}
    q\subs\qfresh\flt
  }{
    \G[\flt]\ts t : \ty[q]{T}
  }
  \\
  \infrule[\ruledef{tref}{t-ref}]{
    \phantom{text} \\
    \G[\flt]\ts t : \ty[q]{T} \andalso
    \qfresh \notin q
  }{
    \G[\flt]\ts \tref{t} : \TRef[\!\qfresh]{\ty[q]{T}}{}
  }
  &
  \infrule[\ruledef{tderef}{t-get}]{
    \G[\flt]\ts t : \TSRef[p]{h}{P}{\ty[q]{T}} \\
    q\subs h,\flt \andalso
    \qfresh\notin p \lor h \notin T
  }{
    \G[\flt]\ts \tget{t} : \ty[q]{T} [p/h]
  }
  &
  \infrule[\ruledef{tassign}{t-put}]{
    \G[\flt]\ts t_1 : \TSRef[p]{h}{P}{Q} \\
    \G[\flt]\ts t_2 : P
  }{
    \G[\flt]\ts \tput{t_1}{t_2} : \TBase[\qbot]
  }
\end{tabular} \\
\begin{tabular}{@{}p{.5\linewidth}@{}p{.5\linewidth}@{}}
  \infrule[\ruledef{tabs}{t-abs}]{
    \phantom{text} \\
    \pcx[q,f,x]{\G,\ \ctxvar{f}{\TTop[q]},\ \ctxvar{x}{\ty[p]{T}}}
      \ts t : Q \\
    p \subs \qfresh{q} \andalso
    q \subs \flt
  }{
    \G[\flt] \ts \tlam{f}{x}{t} : \TFun[q]{f}{x}{\ty[p]{T}}{Q}
  }
  &
  \infrule[\ruledef{tapp}{t-app}]{
    \G[\flt]\ts t_1 : \TFun[q]{f}{x}{\ty[p]{T}}{\ty[r]{U}} \\
    \G[\flt] \ts t_2 : \ty[s]{T} \andalso
    \G[\flt] \ts s \subapp{q} p \\
    r \subs \qfresh \flt,f,x \phantom{mn}
    \qfresh \notin s \lor x \notin U \phantom{mn}
    \qfresh \notin q \lor f \notin U
  }{
    \G[\flt]\ts \tapp{t_1}{t_2} : \ty[r]{U}[s/x, q/f]
  }
  \\
  \infrule[\ruledef{ttabs}{t-tabs}]{
    \phantom{text} \\
    \pcx[q,f,x]{\G,\ \ctxvar{f}{\TTop[q]},\ \ctxtyp{X}{\ty[p]{T}}}
      \ts t : \ty{Q} \\
    p \subs \qfresh{q} \andalso
    q \subs \flt
  }{
    \G[\flt] \ts \ttlam{f}{X}{t} : \TAll[q]{f}{X}{\ty[p]{T}}{Q}
  }
  &
  \infrule[\ruledef{ttapp}{t-tapp}]{
    \G[\flt]\ts t : \TAll[q]{f}{X}{\ty[p]{T}}{\ty[r]{U}} \\
    \G \ts V \subq T \phantom{mn}
    s \subs \qfresh\flt \andalso
    \G[\flt] \ts s \subapp{q} p \\
    r \subs \qfresh \flt,f,x \phantom{mn}
    \qfresh \notin s \lor x \notin U \phantom{mn}
    \qfresh \notin q \lor f \notin U
  }{
    \G[\flt]\ts \ttapp{t}{\ty[s]{V}} : \ty[r]{U}[\ty[s]{V}/\tqvar{X}, q/f]
  }
\end{tabular} \\[1em]

\judgement{Application Conformance}{\BOX{\G[\flt] \ts q \subapp{q} q}} \\[-1em]
\begin{tabular}{@{}p{.29\linewidth}@{}p{.71\linewidth}@{}}
  \infrule[\ruledef{fsub}{f-sub}]{
    \G\ts s <: p
  }{
    \G[\flt] \ts s \subapp{q} p
  }
  &
  \infrule[\ruledef{ffresh}{f-fresh}]{
    \G \ts s \overlap q \subq \qfresh{p} \andalso
    s \overlap q \subs \qfresh{\flt} \andalso
    {\qhole{} \notin \qsat{s},\qsat{q}}
  }{
    \G[\flt] \ts s \subapp{q} \qfresh{p}
  }
\end{tabular} \\[-1ex]

\caption{Select typing rules of \maybelang.
Adapted from \citet{DBLP:journals/pacmpl/WeiBJBR24},
the rules are presented with explicit conformance to merge their
two separate application rules, and with shallow, dual-component reference types \cite{DBLP:journals/pacmpl/0001HJBR25}.
Rules for type annotations are available
in \Cref{fig:decpoly}. } \label{fig:typing}
\end{figure} 
\Cref{fig:typing} presents the typing rules of
$\maybelang$, written in the form $\G[\flt] \ts t : Q$, where $\flt$ is the
\emph{observation filter}, summarizing all free variables required to type the
term $t$.
The typing rules mainly follow the design of prior works
\cite{DBLP:journals/pacmpl/WeiBJBR24,DBLP:journals/pacmpl/0001HJBR25}, but
build on the subtyping rules in \Cref{sec:subtyping}.
We use $\subt$ to denote the new combined type-and-qualifier subtyping relations,
and $\subq$ for the separate relations of qualifiers and types,
as appeared in prior work \cite{DBLP:journals/pacmpl/WeiBJBR24}.

Basic typing rules handle constants and variables. For constants
\rulename{tcst}, the empty qualifier is assigned, as primitive values do not
track resources. For the variable $x$ \rulename{tvar}, the qualifier is $x$,
regardless of the qualifier $q$ recorded in the context. This $q$ can later be
revealed via subsumption; see \rulename{qvar} in \Cref{fig:subtype}. Additionally, $x$ must appear in the
observation $\flt$.

Allocation \rulename{tref} yields a reference shallowly qualified by only
$\qfresh$.
Aligned with prior work~\cite{DBLP:journals/pacmpl/WeiBJBR24,DBLP:journals/pacmpl/BaoJ0BR25},
the referent must be non-fresh.
Both components in the resulting type are the same, thus abbreviated.
Dereferencing \rulename{tderef} replaces the self-reference $h$ in the
\emph{get} component with the reference qualifier itself.
Such substitution of bound variables observes the same restriction seen in
\Cref{sec2:depfun} and has to be
either non-fresh ($\qfresh\notin p$) or non-deep ($h \notin T$).
Assigning the reference \rulename{tassign} concerns the \emph{put}
component and is otherwise standard.

Subsumption \rulename{tsub} and ascription are the
foundation for mode switching in bidirectional typing. In \rulename{tsub}, we
apply extended subtyping ($\subt$) to enable expressive conversions involving
self-references. The qualifier $q$ of the supertype must be bounded
by the filter $\flt$.
The ascription rule and the rules for abstractions with
domain annotations
trivially delegate to unannotated terms and are given in \Cref{fig:decpoly}.

The abstraction rule \rulename{tabs} introduces both the self-reference $f$
and the argument variable~$x$. Representing only reachability,
the self-reference $f$ is given the top type with the function qualifier~$q$.
This $q$ extended with $f$ and $x$ then restricts the observation for the body $t$.

The application rule \rulename{tapp} requires the parameter and the argument
to have the same type $T$, but allows different qualifiers.
Reflecting the two separate application rules in prior works \cite{DBLP:journals/pacmpl/WeiBJBR24,DBLP:journals/pacmpl/0001HJBR25},
the parameter qualifier $s$ must either be
bounded by $p$ \rulename{fsub}, or
overlap with the function qualifier $q$ by no more than $p$ \rulename{ffresh}.
Holes are rejected in saturation computation, ensuring that future
instantiations do not affect the result.
Besides the two cases, we mechanize an extension (\Cref{sec:arbitraryarg}, \cite{artifact})
that may accept parameters with arbitrary reachability.
To complete the application, we substitute $f$ and $x$ in the result type with
$q$ and $s$, respectively. Similar to the case of \rulename{tderef},
such substitutions have to be either non-fresh or non-deep, according to
\Cref{sec2:depfun}.

Type abstractions \rulename{ttabs} and applications \rulename{ttapp}
support bounded quantification. A bound of the form
$\ctxtyp{X}{\ty[p]{T}}$ can be read as a combined type bound $X <: T$ and
qualifier bound $x <: p$, where $X$ and $x$ may be used independently.
These rules parallel their function counterparts.
In \rulename{ttapp}, the type argument $V$ may differ from the bound $T$;
their subtyping is checked using the basic relation ($\subq$), independent of qualifiers.

\subsection{Subtyping and Subqualifying} \label{sec:subtyping}

\begin{figure}[t]
\footnotesize \makeatother
\typicallabel{sq} \setlength{\afterruleskip}{\smallskipamount}

\judgement{Subqualifying}{\BOX{\G\ts q \subq q}} \\[-1em]
\begin{tabular}{@{}p{.28\linewidth}@{}p{.36\linewidth}@{}p{.36\linewidth}@{}}
  \infrule[\ruledef{qsub}{q-sub}]{
    p \subs q
  }{
    \G\ts p \subq q
  }
  &
  \infrule[\ruledef{qtrans}{q-trans}]{
    \G\ts p \subq q \andalso
    \G\ts q \subq r
  }{
    \G\ts p \subq r
  }
  &
  \infrule[\ruledef{qcong}{q-cong}]{
    \G\ts p \subq r \andalso
    \G\ts q \subq s
  }{
    \G\ts p,q \subq r,s
  }
  \\
  \infrule[\ruledef{qvar}{q-var}]{
    \ctxvar{x}{\ty[q]{T}} \in \G \andalso
    \qfresh \notin q
  }{
    \G\ts x \subq q
  }
  &
  \infrule[\ruledef{qtvar}{q-tvar}]{
    \ctxtyp{X}{\ty[q]{T}} \in \G \andalso
    \qfresh \notin q
  }{
    \G\ts x \subq q
  }
  &
  \infrule[\ruledef{qself}{q-self}]{
    \ctxvar{f}{\ty[q,\qhole{?}]{T}} \in \G
  }{
    \G\ts \qclean{q} \subq f
  }
\end{tabular} \\[1em]

\judgement{Type-and-Qualifier Subtyping}{\BOX{\G\ts Q \subt Q}} \\[-1em]
\begin{tabular}{@{}p{.43\linewidth}@{}p{.57\linewidth}@{}}
  \infrule[\ruledef{srefl}{s-grow}]{
    \G\ts p \subq q
  }{
    \G\ts \ty[p]{T} \subt \ty[q]{T}
  }
  &
  \infrule[\ruledef{strans}{s-trans}]{
    \G\ts P \subt Q \andalso
    \G\ts Q \subt R
  }{
    \G\ts P \subt R
  }
  \\
  \infrule[\ruledef{sref}{s-ref}]{
    \begin{array}{@{}c@{\phantom{mn}}c@{}}
      \pcx[\flt_1]{\G, \ctxvar{h}{\TTop[q]}} \ts Q \subt P &
      \flt_1 \subs q,h \\
      \pcx[\flt_2]{\G, \ctxvar{h}{\TTop[q]}} \ts R \subt S &
      \flt_2 \subs q,h
    \end{array}
  }{
    \G\ts\TSRef[q]{h}{P}{R} \subt \TSRef[q]{h}{Q}{S}
  }
  &
  \infrule[\ruledef{sfun}{s-fun}]{
    \begin{array}{@{}l@{\ }l@{\ }c@{\ }l@{\phantom{mm}}l@{}}
      \pcx[\flt_1]{\G, \ctxvar{f}{\TTop[q]}} &
        \ts & Q & \subt P &
        \flt_1 \subs q,f \\
      \pcx[\flt_2]{\G, \ctxvar{f}{\TTop[q]}, \ctxvar{x}{Q}} &
        \ts & R [\!(x,\flt_1)\!/x] & \subt S &
        \flt_2 \subs q,f,x
    \end{array}
  }{
    \G\ts\TFun[q]{f}{x}{P}{R} \subt \TFun[q]{f}{x}{Q}{S}
  }
\end{tabular} \\
\begin{tabular}{@{}p{.23\linewidth}@{}p{.23\linewidth}@{}p{.54\linewidth}@{}}
  \infrule[\ruledef{stop}{s-top}]{
    \phantom{text} \\ \phantom{text}
  }{
    \G\ts \ty[q]{T} \subt \TTop[q]
  }
  &
  \infrule[\ruledef{stvar}{s-tvar}]{
    \phantom{text} \\
    \ctxtyp{X}{\ty[p]{T}} \in \G
  }{
    \G\ts \ty[q]{X} \subt \ty[q]{T}
  }
  &
  \infrule[\ruledef{sall}{s-all}]{
    \begin{array}{@{}l@{\ }l@{\ }c@{\ }l@{\phantom{mn}}r@{}}
      \pcx[\qbot]{\G, \ctxvar{f}{\TTop[q]}} &
        \ts & Q & \subt P \\
      \pcx[\flt_2]{\G, \ctxvar{f}{\TTop[q]}, \ctxtyp{X}{Q}} &
        \ts & R & \subt S &
        \flt_2 \subs q,f,x
    \end{array}
  }{
    \G\ts\TAll[q]{f}{X}{P}{R} \subt \TAll[q]{f}{X}{Q}{S}
  }
\end{tabular} \\[1em]

\judgement{Subtyping Shorthands}{
  \BOX{\G\ts T \subq T}\ \BOX{\G[\flt]\ts Q \subt Q}} \\[-1em]
\begin{tabular}{@{}p{.43\linewidth}@{}p{.57\linewidth}@{}}
  \infrule[\ruledef{suncond}{s-uncond}]{
    \G\ts \ty[\qfresh]{T} \subt \ty[\qfresh]{U}
  }{
    \G\ts T \subq U
  }
  &
  \infrule[\ruledef{sintn}{s-proxy}]{
    \G\ts \ty[\qfresh]{T} \subt \ty[\qfresh\flt]{U} \andalso
    \G\ts p \subq q \andalso
    \G\ts \flt \subq q
  }{
    \G[\flt] \ts \ty[p]{T} \subt \ty[q]{U}
  }
\end{tabular} \\[-1ex]

\caption{Subsumption rules of \maybelang. Subqualifying rules mainly follow
the prior work \cite{DBLP:journals/pacmpl/WeiBJBR24}, with \rulefmt{q-self}
generalized upon freshness and qualifier holes.
The combined type-and-qualifier subtyping is newly proposed.
} \label{fig:subtype}
\end{figure} 
We present subsumption rules of $\maybelang$ in \Cref{fig:subtype},
including subqualifying rules that mainly follow prior design
and type-and-qualifier subtyping rules introduced in \Cref{sec2:subtyping}.

\paragraph{Subqualifying.}

Qualifiers are sets and thus the subset relation is carried over as
\rulename{qsub}. In addition, \rulename{qtrans} and \rulename{qcong} are defined
in a similar way to the transitivity and congruence rules of subset.
Other rules leverage reachability information from the typing context~$\G$.
Rules \rulename{qvar} and \rulename{qtvar} expand variables by replacing them
with their recorded qualifiers in $\G$, provided these qualifiers
are fully specified, \ie, contain no holes or freshness.
Rule \rulename{qself} introduces self-references to upper-bound the variables
contained in the qualifier of the corresponding self-reference.
Conversely, self-references can be expanded using \rulename{qvar}, as long as
their qualifiers are fully established and non-fresh---conditions not required
in \rulename{qself}.

\paragraph{Combined Type-and-Qualifier Subtyping.}

Individual subtyping judgments may change types and/or qualifiers, and
the transitivity rule \rulename{strans} allows composing such changes.
Rule \rulename{srefl} allows changing qualifiers when types are the same
between the subtype and the supertype; a standard reflexivity rule can be
derived, requiring both types and qualifiers to be the same.
All other rules change only the types, requiring the same
qualifiers on both sides.

As discussed earlier in \Cref{sec2:subtyping}, combined type-and-qualifier
subtyping allows justifying avoidance coercions. For this purpose,
rules \rulename{sref}, \rulename{sfun} and \rulename{sall} add the qualifier $q$
agreed between the subtype and the supertype into the typing context, as the
qualifier of their corresponding self-references.
Together with the subqualifying rule for self-references \rulename{qself},
this enables expressing the basic avoidance subtyping 
examples in \Cref{sec2:subtyping}.

When deriving subtyping for the component types, rules
\rulename{sref}, \rulename{sfun} and \rulename{sall} use the auxiliary form
\rulename{sintn}:
they understand the qualifier $p$ of $T$ as \emph{some reachability},
and $q$ of $U$ as \emph{some reachability including} $\flt$, where
the freshness markers do not mean separation.
This way, in \rulename{sref} and \rulename{sfun}, we require 
only $\flt_1,\flt_2$ to be within $q$ modulo bound variables, but
do not restrict the qualifiers of $P,Q,R,S$;
this is crucial to precise, shallow reference types \cite{DBLP:journals/pacmpl/0001HJBR25}, where referent qualifiers are not
necessarily smaller than outer ones.
In \rulename{sfun}, when checking subtyping between the domain types and
qualifiers $R$ and $S$, it uses the substitution
$R[(x,\flt_1)/x]$ to account for the fact that $x$ in $R$ and $S$ may
refer to arguments with reachability differing by $\flt_1$.

We further include \rulename{stop}, \rulename{stvar},
and \rulename{sall} for type polymorphism.
The former two rules are standard.
Rule \rulename{sall} behaves similarly to \rulename{sfun}, except that the
subtyping between type-and-qualifier bounds via \rulename{sintn} requires no
observable $\flt_1$, so that type bounds and qualifier bounds can be used
orthogonally.
This aligns with type applications \rulename{ttapp}, where we use the type-only
subtyping (by~$\subq$) that applies regardless of qualifiers.
In \maybelang, such type-only subtyping is not a separate set of rules, but
can be derived from type-and-qualifier subtyping via \rulename{suncond}, where
both sides agree on the opaque qualifier $\qfresh$.
This is also how subtyping rules from prior work~\cite{DBLP:journals/pacmpl/WeiBJBR24,DBLP:journals/pacmpl/0001HJBR25}
can be understood in the context of this work.

\subsection{Metatheory} \label{sec3:meta}

\subsubsection{Semantic Soundness.}
Unlike prior work~%
\cite{DBLP:journals/pacmpl/BaoWBJHR21,DBLP:journals/pacmpl/WeiBJBR24,DBLP:journals/pacmpl/0001HJBR25}
that establishes syntactic soundness,
we prove the type soundness of $\maybelang$ using \emph{logical relations}~%
\cite{DBLP:journals/jacm/TimanyKDB24}. The dynamic semantics is formulated as a
big-step interpreter~\cite{DBLP:conf/popl/AminR17}.
Our semantic interpretation is adapted from~\citet{DBLP:journals/pacmpl/BaoJ0BR25},
and we extend their model by adding interpretations for type polymorphism,
shallowly qualified, dual-component references~\cite{DBLP:journals/pacmpl/0001HJBR25},
and deep occurrences of bound variables inside types.
Besides efforts to reflect recent advances in reachability types,
our new subtyping design requires our semantic model to interpret the
reachable locations of values in a type-dependent manner,
and this shift necessitated nontrivial changes.
Details of our logical relations are provided in \Cref{sec:lr},
and we excerpt the key results as follows.

\begin{theorem}[Fundamental]
If a term $t$ is syntactically well-typed, \ie, \:$\G[\flt]\ts t: \ty[q]{T}$, and both the context $\G$ and the store $\sigma$ are well-formed,
then $t$ is also semantically well-typed (\:$\G \models t : \ty[q]{T} $).
Specifically, $t$ evaluates to a value $v$ of the type $\ty{T}$ in finite steps,
such that $v$ may only reach locations described by the qualifier $q$, 
and all store write effects are limited to the locations described by $\flt$.
\end{theorem}

In particular, for terms closed under the empty context, we obtain a formulation
of type safety that requires no assumption about contexts or stores:
\emph{well-typed terms do not get stuck}.
\begin{corollary}[Type Safety]
If $\:\varnothing \ts t : \ty[\varnothing]{T}$, then $t$ evaluates to a value $v$
of type $T$ in finite steps.
\end{corollary}
The \emph{Preservation of Separation} property
\cite{DBLP:journals/pacmpl/WeiBJBR24} also follows from our analysis of store
effects: the evaluation of two well-typed terms with disjoint observations will
observe and update disjoint portions of the store.

\begin{figure}[t]
\footnotesize\makeatother \setlength{\afterruleskip}{\smallskipamount}
\judgement{Context Subsumption}{\BOX{\strut \G\sqsubseteq\G}}\\[-1em]
\begin{tabular}{@{}p{.16\linewidth}@{}p{.2\linewidth}@{}p{.3\linewidth}@{}p{.34\linewidth}@{}}
  \infrule{\quad}{\qbot \sqsubseteq \qbot}
  &
  \infrule{\G\sqsubseteq\G'}
    {\G, \ctxvar{x}{Q}\ \sqsubseteq\ \G', \ctxvar{x}{Q}}
  &
  \infrule{\G\sqsubseteq\G'}
    {\G, \ctxtyp{X}{Q}\ \sqsubseteq\ \G', \ctxtyp{X}{Q}}
  &
  \infrule{\G\sqsubseteq\G' \hspace{3em}
    \G\ts q \hspace{3em}
    \qfresh\notin q}
    {\G, \ctxvar{f}{\ty[p,\qhole]{T}}\ \sqsubseteq\ \G', \ctxvar{f}{\ty[p,q,\qhole]{T}}}
\end{tabular} \\[-1ex]
\caption{Definitions of context subsumption: context entries are either the 
same, or with their qualifier holes instantiated by well-formed,
non-fresh qualifiers. }
\label{fig:alggrow}
\end{figure} 
\subsubsection{Interacting with Qualifier Holes.}

Although \maybelang includes no rules that introduce
qualifier holes, we show that inserting and instantiating such holes
preserves the soundness results.
We define \emph{context subsumption} in \Cref{fig:alggrow} to characterize
the effect of hole instantiation. The relation
$\G \sqsubseteq \G'$ indicates that $\G'$ results from partially instantiating holes
in $\G$ zero or more times.

\begin{lemma}[Context Subsumption on Typing] \label{thm:ctxgrow_on_hastype}
If $\G[\flt] \ts t : Q$, and $\G \sqsubseteq \G'$, then $\cx[\flt]{\G'} \ts t : Q$.
\end{lemma}
\begin{lemma}[Hole Sealing on Typing] \label{thm:gstighten_hastype}
If $\pcx[\flt]{\G_1,\ctxvar{f}{\ty[q,\qhole]{T}},\G_2} \ts t : Q$, then $\pcx[\flt]{\G_1,\ctxvar{f}{\ty[q]{T}},\G_2} \ts t : Q$.
\end{lemma}

These lemmas show that proving a type judgment for a specific context~$\G$
can be reduced to proving it under a \emph{weaker} context--one containing
more holes or instantiations with smaller qualifiers. This property is
crucial for the soundness of qualifier inference in the algorithm. %
\section{$\algolang$: Bidirectional Typing with Qualifier Inference and Avoidance} \label{sec:algo}

The prior presentation of $\maybelang$ is \textit{declarative},
not specifying which qualifiers can be inferred and by what means.
In this section, we introduce its typing algorithm, \algolang.
Following bidirectional typing~\cite{DBLP:journals/csur/DunfieldK21},
\algolang infers both types and qualifiers for terms, given annotations
on function arguments and explicit instantiations of type abstractions.
When a fresh value escapes within other resources, the algorithm
automatically applies avoidance conversions
to track the freshness via self-references.
We have implemented \algolang in Lean and proven
its soundness with respect to \maybelang.
We also prove its termination, at the cost of being incomplete and rejecting 
some valid \maybelang terms.

We organize the algorithm presentation bottom-up,
covering qualifiers (\Cref{sec4:qualifier}),
subtyping (\Cref{sec4:subtype}),
avoidance (\Cref{sec4:avoidance}),
and finally bidirectional typing (\Cref{sec4:bidirectional}).
We briefly discuss the metatheoretical properties of \algolang in \Cref{sec:algometa}.

\subsection{Qualifier Checking and Inference} \label{sec4:qualifier}

\begin{figure}[t]
\footnotesize\makeatother \typicallabel{qd}
\setlength{\afterruleskip}{\smallskipamount}

\judgement{Qualifier Exposure}{
  \BOX{\strut\G\ts q \qexp \algout{q}}
  \ \BOX{\strut\G\ts q \qexp[1] \algout{q}}
  \ \BOX{\strut\G\ts q \qexp[2] \algout{q}}}\\[-1em]
\begin{tabular}{@{}p{\linewidth}@{}}
  \infrule[\ruledef{qejoin}{qe-join}]{
    \G \ts q  \qexp[1]^* \algout{q_1} \hspace{4em}
    \G \ts q_1 \qexp[2]^* \algout{q_2}
  }{
    \G \ts q \Uparrow \algout{q_2}
  }
\end{tabular} \\
\begin{tabular}{@{}p{.4\linewidth}@{}p{.6\linewidth}@{}}
  \infrule[\ruledef{qeself}{qe-self}]{
    \ctxvar{f}{\ty[p,\qhole{?}]{T}} \in \G \andalso
    \qclean{p}\, \nsubseteq q
  }{
    \G \ts f, q\ \qexp[1]\ \algout{f, q, (\qclean{p})}
  }
  &
  \infrule[\ruledef{qevar}{qe-var}]{
    \ctxvar{x}{\ty[p]{T}} \in \G \lor \ctxtyp{X}{\ty[p]{T}} \in \G \andalso
    \qfresh \notin p \andalso
    x \notin q
  }{
    \G \ts p, q\ \qexp[2]\ \algout{p, q, x}
  }
\end{tabular} \\[1em]

\judgement{Qualifier Unification}{\BOX{\strut\G\ts q\qunif q \ots{\G}}}\\[-1em]
\begin{tabular}{@{}p{\linewidth}@{}}
  \infrule[\ruledef{quself}{qu$_1$-unify}]{
    x \notin p,q \andalso
    \begin{array}{l@{\ }l@{\,}l}
    & \G = .., \ctxvar{x}{Q}, & .., \ctxvar{f}{\ty[r,\qhole]{T}}, .. \\
    \lor & \G = .., \ctxtyp{X}{Q}, & .., \ctxvar{f}{\ty[r,\qhole]{T}}, .. \\
    \end{array} \andalso
    \G\ts p \qunif q,f \ots{\G_1,\ctxvar{f}{\ty[s,\qhole]{T}},\G_2}
  }{
    \G \ts p,x \qunif q,f \ots{\G_1,\ctxvar{f}{\ty[s,x,\qhole]{T}},\G_2}
  }
\end{tabular} \\
\begin{tabular}{@{}p{.75\linewidth}@{}p{.25\linewidth}@{}}  
  \infrule[\ruledef{quvar}{qu$_2$-upcast}]{
    x \notin p,q \andalso
    \ctxvar{x}{\ty[r]{T}} \in \G \lor \ctxtyp{X}{\ty[r]{T}} \in \G \andalso
    \qfresh{} \notin r \andalso
    \G\ts p,r \qunif q \ots{\G'}
  }{
    \G\ts p,x \qunif q \ots{\G'}
  }
  &
  \infrule[\ruledef{qubot}{qu$_3$-bot}]{
    p \subs q
  }{
    \G\ts p \qunif q \ots{\G}
  }
\end{tabular} \\[1em]

\judgement{Qualifier Checking and Inference}{\BOX{\strut\G \ts q \subq[] q}
                                           \ \BOX{\strut\G \ts q \subq[] q \ots{\G}}}\\[-1em]
\begin{tabular}{@{}p{.5\linewidth}@{}p{.5\linewidth}@{}}
  \infrule[\ruledef{qcheck}{qcheck}]{
    \G\ts q \qexp \algout{q'} \hspace{4em}
    p \subs q'
  }{
    \G \ts p \subq[] q
  }
  &
  \infrule[\ruledef{qinfer}{qinfer}]{
    \G\ts q \qexp \algout{q'} \hspace{4em}
    \G \ts p \qunif q' \ots{\G'}
  }{
    \G \ts p \subq[] q \ots{\G'}
  }
\end{tabular} \\[-1ex]
\caption{Qualifier checking and inference in \algolang.
In the order of the indices, rules of the same form are tried sequentially,
and the first succeeding one is applied.
Outputs in the rules are marked in \algout{purple}.} \label{fig:algqual}
\end{figure} 
\Cref{fig:algqual} presents the rules for qualifier checking and inference
in \algolang, corresponding to the subqualifying rules in \Cref{fig:subtype}.
These rules can be divided into two main components:
\emph{qualifier exposure} without considering qualifier holes,
and \emph{unification} that instantiates holes.

\paragraph{Qualifier Exposure and Checking}

Qualifier exposure ($\qexp$) for $q$ aims to find a large enough
qualifier $q'$, so that checking $p \subq q$ can be reduced to simply checking
$p \subs q'$, as seen in \rulename{qcheck}.

The exposure procedure \rulename{qejoin} proceeds in two stages.  In both
stages, we extend the input qualifier by a subqualifier.
The first stage
$\qexp[1]$ \rulename{qeself} enumerates self-references in the input qualifier,
and extends the input qualifier with the qualifiers of the found
self-references, reflecting the declarative rule \rulename{qself}.  The second
stage $\qexp[2]$ \rulename{qevar} then adds all variables whose reachability is
already included in the input qualifier, reflecting the declarative rule
\rulename{qvar}.

Although applied iteratively, neither step diverges.
For well-formed contexts,
the first stage can finish by scanning the context once in the reverse order,
while the second stage can finish by scanning once in the forward order.

\paragraph{Qualifier Unification and Inference}

Qualifier checking and inference differ in their approach to constraint satisfaction.
While qualifier checking ignores all qualifier holes, qualifier inference \rulename{qinfer} replaces
subset checking~$\subs$ with qualifier unification~$\qunif$, eagerly satisfying constraints by
instantiating qualifier holes in the context.

Qualifier unification concludes by \rulename{qubot} when $p$ is simply a subset
of $q$. Whenever this is not the case, there must be an outstanding variable
$x$ that does not appear in $q$.
Primarily, the unification rule \rulename{quself} tries to insert $x$ into the
qualifier hole of a self-reference $f$.
This unification step is the key to satisfying @b < f@ in the example from
\Cref{sec2:inferqual}.

Crucially, the choice of $f$ is not arbitrary:
(1) $f$ must be defined after $x$,
so that the instantiation yields a well-formed context,
and (2) among all candidates, we select the one defined earliest to
avoid cascading updates on others.
If no such $f$ exists, we apply \rulename{quvar} to replace $x$ with its
reachability $r$.
This is allowed only if $r$ contains no holes or freshness; otherwise, unification fails.

Unification processes all relevant context entries in reverse order,
each exactly once.
To infer qualifiers in subtyping and typing as illustrated in \Cref{sec2:inferqual},
we use \rulename{qinfer} but not \rulename{qcheck}.

\subsection{Subtype Checking} \label{sec4:subtype}

\begin{figure}[t]
\footnotesize\makeatother \typicallabel{sa}
\setlength{\afterruleskip}{\smallskipamount}

\judgement{Toplevel Subtype Checking}{
  \BOX{\strut\G \ts Q \subt[] \ty[\algout{q}]{T} \ots{\G}}}\\[-2em]
\begin{tabular}{@{}p{\linewidth}@{}}
  \infrule[\ruledef{sajoin}{sa-join}]{
    T_1 \subt[1]^{p} \algout{T_1'} \andalso \andalso
    \G\ts T_1' \subt[2]^{p} T_2 \growth{q} \ots{\G'}
  }{
    \G\ts \ty[p]{T_1} \subt[] \ty[\algout{p,q}]{T_2} \ots{\G'}
  }
\end{tabular} \\[1ex]

\judgement{Self Unpacking (Selection)}{
    \BOX{\strut \ty{T} \subt[1]^{q} \algout{\ty{T}}}}\\[-1em]
\begin{tabular}{@{}p{.41\linewidth}@{}p{.41\linewidth}@{}p{.18\linewidth}@{}}
  \infrule[\ruledef{suref}{su$_1$-ref}]{
    \qfresh{} \notin q \andalso \theta = [q/h]
  }{
    \TSRef{h}{\ty[p]{T}}{Q} \,\subt[1]^{q}\,
    \algout{\TSRef{h}{\ty[p]{T\!\theta}}{Q\!\theta}}
  }
  &
  \infrule[\ruledef{sufun}{su$_2$-fun}]{
    \qfresh{} \notin q \andalso \theta = [q/f]
  }{
    \TFun{f}{x}{\ty[p]{T}}{Q} \,\subt[1]^{q}\,
    \algout{\TFun{f}{x}{\ty[p]{T\!\theta}}{Q\!\theta}}
  }
  &
  \infrule[\ruledef{subot}{su$_4$-bot}]{\quad}{
    T \subt[1]^{q} \algout{T}
  }
\end{tabular} \\[1em]

\judgement{Recursive Subtype Checking}{
    \BOX{\strut\G \ts \ty{T} \subt[2]^{q} \ty{T} \growth{q} \ots{\G}}}\\[-1ex]
\begin{tabular}{@{}p{.28\linewidth}@{}p{.72\linewidth}@{}}  
  \infrule[\ruledef{sabase}{sa-base}]{\quad\\\quad\\\quad}{
    \G \ts B \subt[2]^{q} B \growth{\qbot} \ots{\G}
  }
  &
  \infrule[\ruledef{saref}{sa-ref}]{
    \G, \ctxvar{h}{\TTop[q,\qhole{}]} \ts T_2
      \subt[2]^{\qfresh} T_1 \growth{\flt_1} \ots{\G_1} \andalso
    \G_1 \ts p_2,\flt_1 \subq[] p_1 \ots{\G_2} \hspace{4em} \\ \hspace{4em}
    \G_2 \ts U_1
      \subt[2]^{\qfresh} U_2 \growth{\flt_2} \ots{\G_3} \andalso
    \G_3 \ts r_1, \flt_2 \subq[] r_2
      \ots{\G', \ctxvar{h}{\TTop[q',\qhole]}} \\
    q'' = (q'\!\setminus\!q),\, (\flt_1\!\setminus\!h),\, (\flt_2\!\setminus\!h)
  }{
    \G\ts
    \TSRef{h}{\ty[p_1]{T_1}}{\ty[r_1]{U_1}}
    \:\subt[2]^{q}\:
    \TSRef{h}{\ty[p_2]{T_2}}{\ty[r_2]{U_2}} \growth{q''} \ots{\G'}
  }
\end{tabular} \\
\begin{tabular}{@{}p{.28\linewidth}@{}p{.28\linewidth}@{}p{.44\linewidth}@{}}  
  \infrule[\ruledef{satop}{sa-top}]{\quad}{
    \G \ts T \subt[2]^{q} \TTop \growth{\qbot} \ots{\G}
  }
  &
  \infrule[\ruledef{satrefl}{sa-tvar$_1$}]{\quad}{
    \G \ts X \subt[2]^{q} X \growth{\qbot} \ots{\G}
  }
  &
  \infrule[\ruledef{sattrans}{sa-tvar$_2$}]{
    \ctxtyp{X}{\ty[p]{U}} \in \G \andalso
    \G \ts U \subt[2]^{q} T \growth{q'} \ots{\G'}
  }{
    \G \ts X \subt[2]^{q} T \growth{q'} \ots{\G'}
  }
\end{tabular} \\
\begin{tabular}{@{}p{\linewidth}@{}}
  \infrule[\ruledef{safun}{sa-fun}]{
    \G, \ctxvar{f}{\TTop[q,\qhole{}]} \ts T_2
      \subt[2]^{\qfresh} T_1 \growth{\flt_1} \ots{\G_1} \andalso
    \G_1 \ts p_2,\flt_1 \subq[] p_1 \ots{\G_2} \andalso
    \theta = [(x,\flt_1)/x] \\
    \G_2, \ctxvar{x}{\ty[p_2]{T_2}} \ts U_1\!\theta
      \subt[2]^{\qfresh} U_2 \growth{\flt_2} \ots{\G_3} \andalso
    \G_3 \ts r_1\!\theta, \flt_2 \subq[] r_2
      \ots{\G', \ctxvar{f}{\TTop[q',\qhole]},..} \hspace{3em} \\
    q'' = (q'\!\setminus\!q),\, (\flt_1\!\setminus\!f),\, (\flt_2\!\setminus\!\{f,x\})
  }{
    \G\ts
    \TFun{f}{x}{\ty[p_1]{T_1}}{\ty[r_1]{U_1}}
    \:\subt[2]^{q}\:
    \TFun{f}{x}{\ty[p_2]{T_2}}{\ty[r_2]{U_2}} \growth{q''} \ots{\G'}
  }
  \\
  \infrule[\ruledef{saall}{sa-all}]{
    \G, \ctxvar{f}{\TTop[q,\qhole{}]} \ts p_2 \subq[] p_1 \ots{\G_2} \\
    \G_2, \ctxtyp{X}{\ty[p_2]{T}} \ts U_1
      \subt[2]^{\qfresh} U_2 \growth{\flt_2} \ots{\G_3} \andalso
    \G_3 \ts r_1, \flt_2 \subq[] r_2
      \ots{\G', \ctxvar{f}{\TTop[q',\qhole]}, \cdots} \\
    q'' = (q'\!\setminus\!q),\, (\flt_2\!\setminus\!\{f,x\})
  }{
    \G\ts
    \TAll{f}{X}{\ty[p_1]{T}}{\ty[r_1]{U_1}}
    \:\subt[2]^{q}\:
    \TAll{f}{X}{\ty[p_2]{T}}{\ty[r_2]{U_2}} \growth{q''} \ots{\G'}
  }
\end{tabular} \\[-1ex]

\caption{Select subtype checking in \algolang.
Additional rules are available in \Cref{fig:subcheckmore}.
Rules are syntax-directed or ordered by indices,
with outputs marked in \algout{purple}.
} \label{fig:subcheck}
\end{figure} 
\Cref{fig:subcheck} presents the rules for subtype checking,
which also infers the qualifier for the supertype.
The top-level procedure \rulename{sajoin} operates in two phases:
\emph{self unpacking} and \emph{recursive checking}.
We write $\subt[]$ for the algorithm,
in contrast to the declarative~$\subt$.

\paragraph{Recursive Subtype Checking}

The second phase $\subt[2]$ adapts the declarative subtyping rules.
Here, \rulename{srefl} is specialized into \rulename{sabase} and
\rulename{satrefl}, while \rulename{strans} is internalized in
\rulename{sattrans}.
Rule \rulename{saall} implements a kernel variant of $F_{<:}$ for decidability,
unlike the declarative \rulename{sall} based on the full variant~\cite{DBLP:journals/tcs/Ghelli95}.
Rules \rulename{saref}, \rulename{safun} and \rulename{saall}
initialize the qualifier of their self-references using the input qualifier
$q$ with a hole $\qhole$.
With \rulename{qinfer}, they eagerly satisfy constraints on the hole
(\Cref{sec2:inferqual}) and thus infer the additional reachability
$\algout{q''}$ to appear in the supertype qualifier.
This phase alone enables the \emph{packing} conversion illustrated
in \Cref{sec2:subtyping}.

\paragraph{Self Unpacking}

With qualifier holes in self-reference qualifiers, the recursive phase does not
support the \emph{unpacking} conversion from \Cref{sec2:subtyping}, as
\rulename{qvar} requires the fully specified qualifiers.
To address this, we add a dedicated unpacking step
prior to recursive checking, denoted $\subt[1]$,
to replace $f$ with the input qualifier if it is not fresh.
The two steps connect by transitivity \rulename{strans}.

\subsection{Avoidance Conversion} \label{sec4:avoidance}

\begin{figure}[t]
\footnotesize\makeatother \typicallabel{sa}
\setlength{\afterruleskip}{\smallskipamount}

\judgement{Avoidance Core}{%
    \BOX{\strut Q \ll_{x} \algout{Q}}} \\[-1em]
\begin{tabular}{@{}p{.33\linewidth}@{}p{.33\linewidth}@{}p{.34\linewidth}@{}}  
  \infrule[\ruledef{avref}{av-ref}]{
    \algout{T_1} := T [h/^-z] \quad
    {p_1} = p \setminus z \\
    \algout{U_1} := U [h/^+z] \quad
    {r_1} = r [h/z] \\
    Q = \TSRef[q,z]{h}{\ty[p_1]{T_1}}{\ty[r_1]{U_1}}
  }{
    \TSRef[q]{h}{\ty[p]{T}}{\ty[r]{U}} \ll_z \algout{Q}
  }
  &
  \infrule[\ruledef{avfun}{av-fun}]{
    \algout{T_1} := T [f/^-z] \quad
    {p_1} = p \setminus z \\
    \algout{U_1} := U [f/^+z] \quad
    {r_1} = r [f/z] \\
    Q = \TFun[q,z]{f}{x}{\ty[p_1]{T_1}}{\ty[r_1]{U_1}}
  }{
    \TFun[q]{f}{x}{\ty[p]{T}}{\ty[r]{U}} \ll_z \algout{Q}
  }
  &
  \infrule[\ruledef{avall}{av-all}]{
    \algout{T_1} := T [f/^-z] \quad
    {p_1} = p \setminus z \\
    \algout{U_1} := U [f/^+z] \quad
    {r_1} = r [f/z] \\
    Q\! =\!\TAll[q,z]{f}{X}{\ty[p_1]{T_1}}{\ty[r_1]{U_1}}
  }{
    \TAll[q]{f}{X}{\ty[p]{T}}{\ty[r]{U}} \ll_z \algout{Q}
  }
\end{tabular} \\[1em]

\judgement{Conditional Avoidance}{%
    \BOX{Q \ll_{q/x} \algout{Q}}
  \ \BOX{Q \ll_{q/x}\ll_{q/x} \algout{Q}}} \\[-1em]
\begin{tabular}{@{}p{.3\linewidth}@{}p{.3\linewidth}@{}p{.4\linewidth}@{}}
  \infrule[\ruledef{acskip}{ac$_1$-skip}]{
    \qfresh \notin q \lor x \notin T
  }{
    \ty[p]{T} \ll_{q/x} \algout{\ty[p]{T}}
  }
  &
  \infrule[\ruledef{acavoid}{ac$_2$-avoid}]{
    Q \ll_x \algout{Q'}
  }{
    Q \ll_{q/x} \algout{Q'}
  }
  &
  \infrule[ac-double]{
    Q \ll_{p/x} \algout{Q'} \andalso Q' \ll_{q/y} \algout{Q''}
  }{
    Q \ll_{p/x}\ll_{q/y} \algout{Q''}
  }
\end{tabular} \\[-1ex]

\caption{Avoidance rules in \algolang.
Rules are syntax-directed or ordered by indices,
with outputs marked in \algout{purple}.
} \label{fig:avoidalg}
\end{figure} 
In \algolang, we implement avoidance conversions as described in
\Cref{sec2:avoid}.
Shown in \Cref{fig:avoidalg}, the core
form \rulefmt{av-$*$} is written as $P \ll_z \algout{Q}$, meaning that avoiding
variable $z$ deep inside the type $P$ yields type $\algout{Q}$.
It is defined by \emph{polarized substitution} (\Cref{fig:avoidmore}):
using the outermost self-reference $f$ in $P$
to replace $z$ in covariant positions, and removing it from contravariant ones.
Since the outermost function subsumes the scope of inner ones,
this choice eliminates the need to handle inner self-references
and results in smaller types by subtyping.

\Cref{fig:avoidalg} also defines \emph{conditional avoidance} \rulefmt{ac-$*$},
invoking avoidance core only on restricted substitutions.
Note that we do not define avoidance core
for simple types: there, the occurrence requirement $x \notin T$
in \rulename{acskip} is trivially true, and thus they are never required for
\rulename{acavoid}.

\subsection{Bidirectional Typing} \label{sec4:bidirectional}

\begin{figure}[t]
\footnotesize\makeatother \typicallabel{tiapp}
\setlength{\afterruleskip}{\smallskipamount}
\judgement{Bidirectional Typing (Selection)}{
    \BOX{\strut\G[\algout{\flt}] \ts t \synth\![\qexp]\, \algout{\ty{Q}} \ots{\G}}
  \ \BOX{\strut\G[\algout{\flt}] \ts t \chek \ty[\synth \algout{q}]{T} \ots{\G}}
  \ \BOX{\strut\G[\algout{\flt}] \ts t \chek \ty{Q} \ots{\G}}}\\[-1ex]
\begin{tabular}{@{}p{.3\linewidth}@{}p{.3\linewidth}@{}p{.4\linewidth}@{}}
  \infrule[\ruledef{ticst}{ti$_1$-cst}]{
    c \in \TBase
  }{
    \G[\algout{\qbot}] \ts c \synth \algout{\TBase[\qbot]} \ots{\G}
  }
  &
  \infrule[\ruledef{tivar}{ti$_2$-var}]{
    \ctxvar{x}{\ty[q]{T}} \in \G
  }{
    \G[\algout{x}] \ts x \synth \algout{\ty[x]{T}} \ots{\G}
  }
  &
  \infrule[\ruledef{tiref}{ti$_3$-ref}]{
    \G[\algout{\flt}] \ts t \synth \algout{\ty[q]{T}} \ots{\G'} \andalso
    \qfresh \notin q
  }{
    \G[\algout{\flt}] \ts \tref{t} \synth \algout{\TRef[\qfresh]{\ty[q]{T}}{}} \ots{\G'}
  }
\end{tabular}
\begin{tabular}{@{}p{.5\linewidth}@{}p{.5\linewidth}@{}}
  \infrule[\ruledef{tideref}{ti$_4$-get}]{
    \G[\algout{\flt}] \ts t \synth\qexp \algout{\TSRef[r]{h}{P}{Q}} \ots{\G'} \\
    Q \ll_{r/h} \algout{\ty[q]{U}} \andalso
    \flt' = \flt, (q\!\setminus\!h)
  }{
    \G[\algout{\flt'}] \ts \tget{t} \synth \algout{\ty[q]{U} [r/h]} \ots{\G'}
  }
  &
  \infrule[\ruledef{tiassign}{ti$_5$-put}]{
    \G[\algout{\flt_1}] \ts t \synth\qexp \algout{\TSRef[r]{h}{\ty[p]{T}}{Q}} \ots{\G_1} \\
    \qfresh\notin r \lor h\notin T \andalso
    \G[\algout{\flt_2}]_1 \ts t_2 \chek \ty[p]{T} [r/h] \ots{\G_2}
  }{
    \G[\algout{\flt_1,\flt_2}] \ts \tput{t_1}{t_2} \synth \algout{\TBase[\qbot]} \ots{\G_2}
  }
  \\
  \infrule[\ruledef{tcabs}{tc$_1$-abs}]{
    \quad\\
    \pcx[\algout{\flt}]{\G, \ctxvar{f}{\TTop[\qhole]}, \ctxvar{x}{\ty[p]{T}}}
      \ts t \chek \ty[q]{U} \ots{\G',\ctxvar{f}{\TTop[q,\qhole]},..} \\
    r = (p,q,\flt)\setminus\{\qfresh{f,x}\}
  }{
    \G[\algout{r}] \ts \tlam{f}{x}{t} \chek \TFun[\synth\algout{r}]{f}{x}{\ty[p]{T}}{\ty[q]{U}} \ots{\G'}
  }
  &
  \infrule[\ruledef{tiabsa}{ti$_6$-abs}]{
    \G,\ctxvar{f}{\TTop[\qfresh]} \ts \ty[p]{T}
    \andalso f \notin^+ T \andalso f \notin p \\
    \pcx[\algout{\flt}]{\G, \ctxvar{f}{\TTop[\qhole]}, \ctxvar{x}{\ty[p]{T}}}
      \ts t \synth \algout{\ty[q]{U}} \ots{\G', \ctxvar{f}{\TTop[q,\qhole]},..} \\
    \algout{V}\!:=\!U[f/^+f] \andalso
    {r} = (p,q,\flt)\setminus\{\qfresh{f,x}\} %
  }{
    \G[\algout{r}] \ts \tlam[{\ty[p]{T}}]{f}{x}{t} \synth \algout{\TFun[r]{f}{x}{\ty[p]{T}}{\ty[q]{V}}} \ots{\G'}
  }
  \\
  \infrule[\ruledef{tilet}{ti$_8$-let}]{
    \quad\\
    \G[\algout{\flt_2}] \ts t_2 \synth \algout{\ty[p]{T}} \ots{\G_1} \\
    \G[\algout{\flt_1}]_1 \ts \tlam[{\ty[p]{T}}]{f}{x}{t_1}
      \synth \algout{\TFun[q]{f}{x}{\ty[p]{T}}{Q}} \ots{\G_2} \\
    Q \ll_{p/x}\ll_{q/f} \algout{\ty[r]{U}} \andalso
    {\flt} = \flt_1,\flt_2,(r\!\setminus\!\{\qfresh{f,x}\})
  }{
    \G[\algout{\flt}] \ts \tapp{(\tlam{f}{x}{t_1})}{t_2} \synth \algout{\ty[r]{U} [p/x,q/f]} \ots{\G_2}
  }
  &
  \infrule[\ruledef{tiapp}{ti$_9$-app}]{
    \G[\algout{\flt_1}] \ts t_1 \synth\qexp
      \algout{\TFun[q]{f}{x}{\ty[p]{T}}{Q}} \ots{\G_1} \\
    \qfresh \notin q \lor f \notin T \andalso
    \G[\algout{\flt_2}]_1 \ts t_2 \chek \ty[\synth\algout{s}]{T[q/f]} \ots{\G_2} \\
    \G[\algout{\flt_3}]_2 \ts s \subapp[]{q} p \ots{\G_3} \\
    Q \ll_{s/x}\ll_{q/f} \algout{\ty[r]{U}} \phantom{mn}
    \flt = \flt_1,\flt_2,\flt_3,(r\!\setminus\!\{\qfresh{f,x}\})
  }{
    \G[\algout{\flt}] \ts \tapp{t_1}{t_2} \synth \algout{\ty[r]{U} [s/x,q/f]} \ots{\G_3}
  }
  \\
  \infrule[\ruledef{tianno}{ti$_{11}$-as}]{
    \G\ts Q \andalso
    \G[\algout{\flt}]\ts t \chek Q \ots{\G'}
  }{
    \G[\algout{\flt}]\ts \tast{t}{Q} \synth \algout{Q} \ots{\G'}
  }
  &
  \infrule[\ruledef{tiexpose}{ti-exp}]{
    \G[\algout{\flt}] \ts t \synth \algout{\ty[q]{T}} \ots{\G'} \andalso
    \G' \ts T \qexp \algout{U}
  }{
    \G[\algout{\flt}] \ts t \synth\qexp \algout{\ty[q]{U}} \ots{\G'}
  }
  \\
  \infrule[\ruledef{tcsub}{tc$_3$-sub}]{
    \G[\algout{\flt}]\ts t \synth \algout{Q} \ots{\G_1} \andalso
    \G_1\ts Q \subt[] \ty[\algout{q}]{T} \ots{\G_2}
  }{
    \G[\algout{\flt,\qclean{q}}] \ts t \chek \ty[\synth\algout{q}]{T} \ots{\G_2}
  }
  &
  \infrule[\ruledef{tqsub}{tq-sub}]{
    \G[\algout{\flt}] \ts t \chek \ty[\synth \algout{q'}]{T} \ots{\G_1} \andalso
    \G_1 \ts q' < q \ots{\G_2}
  }{
    \G[\algout{\flt,\qclean{q}}] \ts t \chek \ty[q]{T} \ots{\G_2}
  }
\end{tabular} \\[1em]

\judgement{Type Exposure}{\BOX{\G\ts T \qexp \algout{T}}} \\[-3ex]
\begin{tabular}{@{}p{.5\linewidth}@{}p{.5\linewidth}@{}}
  \infrule[\ruledef{tutvar}{tu$_1$-tvar}]{
    \ctxtyp{X}{\algout{\ty[q]{T}}} \in \G \andalso
    \G \ts T \qexp \algout{U}
  }{
    \G \ts X \qexp \algout{U}
  }
  &
  \infrule[\ruledef{tubot}{tu$_2$-bot}]{\quad}{
    \G \ts T \qexp \algout{T}
  }
\end{tabular} \\[-1ex]
\caption{Select bidirectional typing rules in \algolang.  Rules in the same form
are ordered by indices, and outputs are marked in \algout{purple}.  Additional
rules are available in \Cref{fig:bidir-suppl}.}
\label{fig:bidir}
\end{figure}
 
We present typing rules for \algolang in \Cref{fig:bidir},
bidirectionalized \cite{DBLP:journals/csur/DunfieldK21} from \Cref{fig:typing}.
It involves inferring~($\synth$) both type and qualifier,
checking the type while inferring a qualifier~($\ty[\synth]{\chek}$),
and checking~($\chek$) both type and qualifier.
As described in \Cref{sec2:inferqual},
all three modes synthesize filters~$\flt$ and
produce output contexts to propagate partially inferred qualifiers.
Mode switching is handled by ascription \rulename{tianno}
and subsumption \rulename{tcsub} \rulename{tqsub}.

The declarative rules \rulename{tcst} \rulename{tvar} \rulename{tref}
become inference rules \rulename{ticst} \rulename{tivar} \rulename{tiref}.
Reference operations are adapted as rules \rulename{tideref}
\rulename{tiassign}, relying on \emph{type exposure} (\qexp) to upcast type
variables when necessary, also defined in \Cref{fig:bidir}.
Both rules require removing self-references in their referent types by
substitution, and \rulename{tideref} thus further involves avoidance.

For unannotated functions, their types must be checked
\rulename{tcabs}, while inference requires annotations \rulename{tiabsa}.
In both cases, function qualifiers are inferred by
collecting body observations and satisfying subtyping constraints via qualifier
holes (\Cref{sec2:inferqual}).
\rulename{tiabsa} further removes contravariant self-references
in the inferred type according to \rulename{cfun}
by polarized substitution $[f/^+f]$.
The rules for type abstractions %
are similar and given in \Cref{fig:bidir-suppl}.

Function application results are always inferred.
The let-binding rule \rulename{tilet} infers the argument type first,
then infers the function type accordingly.
In contrast, the standard application rule \rulename{tiapp} infers the
function type first, then checks the argument type while inferring a qualifier.
It additionally requires type exposure, unpacking self-references
in the codomain, and checking qualifier conformance.
Both rules finish inference with avoidance and qualifier substitution.
The rule for type application %
mirrors \rulename{tiapp} and is given in \Cref{fig:bidir-suppl}.

\subsection{Metatheory} \label{sec:algometa}

For brevity, we leave the detailed metatheory of \algolang to \Cref{app:algo}.
Here, we summarize the key results:
all procedures of \algolang terminate, and they are sound with
respect to the declarative~\maybelang:

\begin{theorem}[Decidability and Soundness of Bidirectional Typing]
  If $\ \WF{\G}$, then it can be decided in finitely many steps whether there exist such $\G'$, $\flt$, $T$, $q$ that $\G[\algout{\flt}] \ts t \Rightarrow \algout{\ty[q]{T}} \dashv \algout{\G'}$, and if so, $\ \G \sqsubseteq \G'$, $\ \G\ts \flt$, $\ \qfresh\notin\flt$, and $\ \cx[\flt]{\G'}\ts t : \ty[q]{T}$.
\end{theorem}

By necessity, \algolang is incomplete.
Since we implement the kernel variant of $F_{<:}$, we cannot accept terms valid
only in the full variant.
Besides, we show that qualifier
unification and thus inference have no principal solution (\Cref{thm:nonprincipal}).
Still, \algolang
produces reasonable results: qualifier exposure and thus \rulename{qcheck} are
proved complete (\Cref{thm:qtp_complete}), self unpacking ($\subt[1]$) results
are proved minimal (\Cref{lemma:self-unpack-equiv}), and the
avoidance strategy using the outermost self-reference is better than other
choices (\Cref{thm:minavoid}).
Given the nontrivial nature of sharing and separation reasoning, as
a type-based approach, we prefer termination over completeness.
In \Cref{sec:eval},
we show that \algolang is expressive enough to check meaningful programs. %
\section{Evaluation and Discussion} \label{sec:eval}

In this section, we evaluate \algolang using programming
examples involving data structures.
Specifically, we show that \algolang is more ergonomic than prior work.
Additionally, we discuss how to extend \algolang with implicit type
instantiation, and compare it with alternative avoidance solutions.

\subsection{Encoding and Natively Supporting Data Structures} \label{sec:vsprior}

Using the essential constructs of lambda calculus,
\algolang can support data structures via Church-encodings.
In addition, we have mechanized support for pairs and lists as examples
of native data types.
Church encodings let us experiment with different reachability
patterns and validate the expressiveness of our avoidance and inference 
mechanisms.
Once validated, we implement the same typing behavior as primitive rules, 
enabling ergonomic interfaces without extensive type annotations.
In this section, we discuss the data structures we support natively or via
encodings.

\paragraph{Escaping Pairs}

Prior work \cite{DBLP:journals/pacmpl/WeiBJBR24} proposed two encoding styles for pairs:
\emph{transparent}, tracking fields as separate components, or \emph{opaque},
without differentiating them.
Illustrated below on the left, the pair is kept transparent within the 
scope of @a,b@, but has to be made opaque when escaping.
Prior work requires a term-level coercion to reconstruct the pair as @OPair@.
After escaping, opaque pairs have to use different eliminators
(@fstO@ instead of @fstT@) to access the fields.

\noindent
\begin{minipage}{.5\linewidth}
\begin{lstlisting}
// prior work: encode Pair,OPair,fstT,fstO,...
val opaque = {
  val a = new Ref(1); val b = new Ref(2)
  val transparent = Pair[...](a, b)
  fstT[...](transparent)  // specialized fstT
  OPair[...](      // require $\eta$-expansion
    fstT[...](transparent),
    sndT[...](transparent))
}                  //: \mup.OPair[..^p^,..^p^]^~*^
fstO[...](opaque)  // specialized fstO
\end{lstlisting}
\end{minipage}
\begin{minipage}{.5\linewidth}
\begin{lstlisting}
// this work: (natively supported)
let opaque = {
  let a = ref 1; let b = ref 2
  let transparent = Pair(a, b) //: Pair[..^a^,..^b^]
  fst(transparent)
  transparent  // automated, seamless avoidance


}              //: \mup.Pair[..^p^,..^p^]^~*^
fst(opaque)    // same eliminator after avoidance
\end{lstlisting}
\end{minipage}
In contrast, %
\algolang (on the right) enables seamless avoidance, both via encodings and natively, that \emph{packs} the
reachability of @a,b@ into the pair.
After avoidance, %
the @opaque@ pair can use the same eliminators as
@transparent@ pairs, without requiring programmers to duplicate eliminators and
differentiate their usages.
Moreover, with native support, instantiation annotations (@[...]@) are not
required.

\paragraph{Lists with Abstract Reachability}

As an extension beyond the formal development of prior work,
we include lists as an example of inductive data types.
Our native list type, @List[T]@, is polymorphic over the element type @T@;
each element is understood to reach the same resources as the list itself.
This interface mirrors Church encodings where element
reachability is tracked using the list's self-reference @h@.
We provide a generic @fold@ operator for implementing other common list operations:

\begin{lstlisting}
// alternative interface: let sum = fun [q](l: List[Ref[Int]^q^]) => ...
let sum = fun (l: \muh.List[Ref[Int]^h^]) = l.fold(0) { i, a => !i + a }
let lst = {
  let a = ref 42; let b = ref 42; let c = ref 42
  sum(a :: nil)                        // List[Ref[Int]^a^]^a^        <= \muh.List[Ref[Int]^h^]^a^
  sum(a :: b :: nil)                   // List[Ref[Int]^a,b^]^a,b^     <= \muh.List[Ref[Int]^h^]^a,b^
  sum(a :: b :: c :: nil)              // List[Ref[Int]^a,b,c^]^a,b,c^  <= \muh.List[Ref[Int]^h^]^a,b,c^
}; sum(lst)                            // List[Ref[Int]^lst^]^lst^      <= \muh.List[Ref[Int]^h^]^lst^
\end{lstlisting}

\paragraph{Lists with Distinct Elements}

As a showcase of the flexibility of our approach,
in \algolang, we can also encode @MList@, where all
elements are guaranteed to be separate
\cite{DBLP:conf/lics/Reynolds02,DBLP:conf/csl/OHearnRY01}.
The cons constructor (@::@, omitting instantiations) enforces head-tail separation.
This separation enables optimizations like iteration reordering in functions such as @miter@.
\begin{lstlisting}
let a = ref 42; let b = ref 42; let c = ref 42
let lst = a :: b :: c :: mnil          // context: [..., lst: MList[Ref[Int]]^a,b,c^]
miter(lst)( fun i => i := !i + 1 )     // can be safely reordered
let lstErr = a :: b :: a :: mnil       // !!Error:!! a is not separate from b::a::mnil
\end{lstlisting}

\subsection{Ergonomics and Performance}

\begin{figure}[b]
  \centering
  \begin{subcaptionblock}{.36\linewidth}
    \makeatother
    \tiny %
\begin{tabular}{l@{\quad}r@{\ }r@{\quad}r@{\ }r@{\ }r} \toprule
& \multicolumn{2}{c}{\textsc{encoded}}	& \multicolumn{3}{c}{\textsc{native}}	\\ \cmidrule(lr){2-3} \cmidrule(lr){4-6}
\textsc{testcase}	&	\textsc{size}	&	\textsc{time (ms)}	&	\textsc{size}	&	\textsc{qual(\%)}	&	\textsc{time (ms)}	\\ \midrule
\textsc{pair-trans}	&	663	&	50.50	&	30	&	0	&	0.25	\\
\textsc{pair-opaque}	&	666	&	45.10	&	33	&	0	&	0.35	\\
\textsc{par-var1}	&	437	&	33.00	&	45	&	10	&	0.90	\\
\textsc{par-var2}	&	437	&	35.40	&	45	&	10	&	0.85	\\
\textsc{par-fun1}	&	322	&	20.65	&	49	&	11	&	1.30	\\
\textsc{par-fun2}	&	363	&	23.35	&	53	&	10	&	1.25	\\
\textsc{par-ref1}	&	446	&	34.95	&	54	&	8	&	1.05	\\
\textsc{par-ref2}	&	446	&	34.70	&	54	&	8	&	1.05	\\
\textsc{seq-ctr}	&	586	&	37.95	&	66	&	3	&	1.50	\\
\textsc{par-ctr1}	&	651	&	42.90	&	65	&	7	&	1.75	\\
\textsc{par-ctr2}	&	606	&	38.00	&	57	&	8	&	1.25	\\
\textsc{par-shr1}	&	409	&	30.00	&	56	&	17	&	1.35	\\
\textsc{par-shr2}	&	409	&	30.20	&	56	&	17	&	1.10	\\
\textsc{withf-esc}	&	132	&	4.00	&	35	&	59	&	0.45	\\
\textsc{withf-seq}	&	599	&	46.55	&	76	&	21	&	1.65	\\
\textsc{withf-par}	&	704	&	61.30	&	103	&	20	&	2.85	\\
\textsc{list-sum}	&	836	&	95.95	&	80	&	3	&	2.15	\\
\textsc{list-map}	&	797	&	66.50	&	94	&	22	&	2.40	\\
\textsc{mlist-sep}	&	606	&	50.15	&	---	&	---	&	---	\\
\textsc{mlist-shr}	&	495	&	38.65	&	---	&	---	&	---	\\
\bottomrule \end{tabular}
     \caption{Example statistics.} \label{tbl:stats}
  \end{subcaptionblock}
  \begin{subcaptionblock}{0.3\linewidth}
    \centering
    \includegraphics[width=\linewidth]{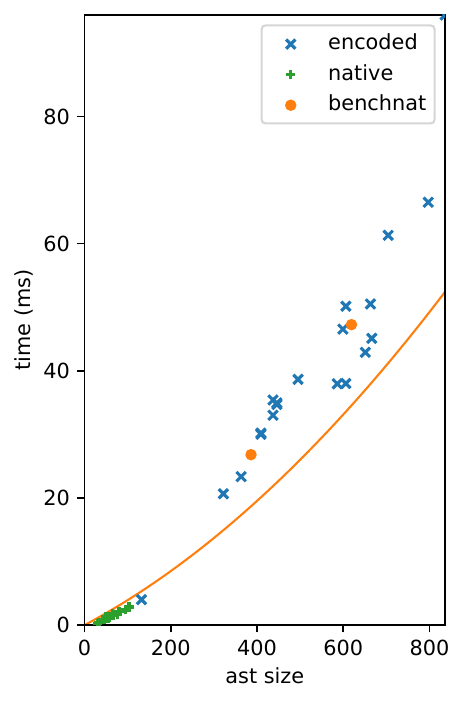}
    \caption{Scaling on examples.} \label{fig:perfmini}
  \end{subcaptionblock}%
  \begin{subcaptionblock}{0.3\linewidth}
    \centering
    \includegraphics[width=\linewidth]{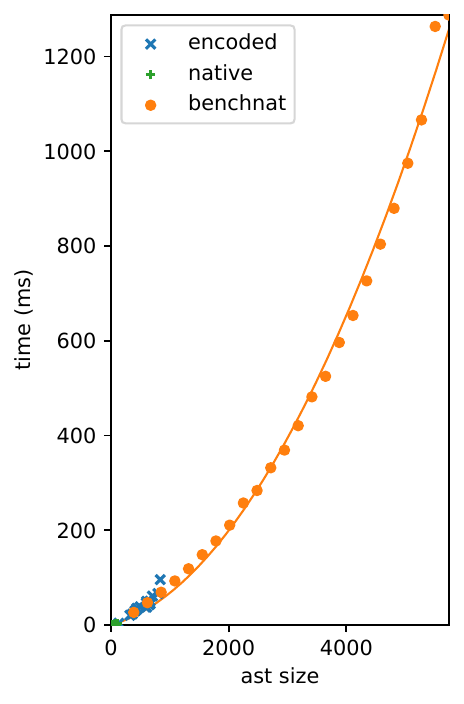}
    \caption{Scaling on \textsc{ChurchNat}.} \label{fig:perfnats}
  \end{subcaptionblock}
  \vspace{-1em}
  \caption{Example statistics and performance scaling results.}
\end{figure}

To validate the ergonomics and performance of \algolang, we implement the
motivating examples 
from \Cref{sec:intro} and variations,
using both Church encodings and native data types.
We refer readers to our artifact \cite{artifact} for full details.

\Cref{tbl:stats} summarizes all examples used for evaluation.
Sizes of examples are measured by the number of AST nodes.
For those implemented using native data types, we further count
\textsc{qual(\%)} as
the percentage of qualifier elements relative to the unqualified program sizes.
For most examples, this is in the range 0--20\%, showing a moderate annotation
burden; the \textsc{withf-esc} example is an outlier, mainly because the
example has simple behavior but carries a nontrivial argument signature.
We note that annotations are concentrated at function definitions; in larger 
programs where each function is called multiple times, the annotation 
percentage would be further amortized.

Time checking each example \rulefmt{time} is measured by the average of
20 consecutive runs.%
\footnote{Measured using Linux 6.6.87.2 on WSL 2.6.3.0 with Intel Ultra 7 258V.}
We plot the checking time of individual examples in \Cref{fig:perfmini}.
To evaluate how \algolang scales on larger inputs, we use synthetic examples
of Church-encoded natural numbers, shown in \Cref{fig:perfnats};
the runtime fits a quadratic trend as both term sizes and context depths increase.
The same quadratic curve, overlaid in \Cref{fig:perfmini}, shows an
imperfect but close match.

\subsection{Implicit Type-and-Qualifier Instantiation}

While \algolang does not infer type-and-qualifier instantiations,
its qualifier inference for functions already adopts an eager instantiation
approach inspired by \citet{DBLP:conf/icfp/DunfieldK13},
resolving lower bound constraints on qualifier holes.
Supporting full instantiation inference would require a more general,
constraint-based approach such as local type inference
\cite{DBLP:journals/toplas/PierceT00},
and we identify two specific challenges.
First, \algolang enjoys \emph{impredicative} polymorphism,
which is beyond the capability of the \emph{complete and easy} approach
\cite{DBLP:conf/icfp/DunfieldK13} and requires more sophisticated solutions
\cite{DBLP:conf/ecoop/ZhaoO22,DBLP:journals/pacmpl/CuiJO23}.
Second, the interleaving of types and qualifiers introduces
upper bound constraints on qualifier variables, requires
eliminating ill-scoped qualifier variables from type constraints,
and demands greatest lower bound computation for qualifiers.
We believe the procedures developed in this work---specifically
avoidance (\Cref{sec4:avoidance}) and qualifier exposure
(\Cref{fig:algqual})---provide reusable foundations towards
generalizing these into a constraint-solving framework.

\subsection{Alternative Avoidance Design} \label{sec:alternative}

Related works that track resources using type qualifiers also observe the
avoidance problem, and we compare them as follows.
Simplistically, some works \cite{DBLP:journals/toplas/BoruchGruszeckiOLLB23,DBLP:journals/pacmpl/GaoP25}
replace variables using \emph{top} qualifiers:

\begin{lstlisting}
{ let b = ref 42; fun () => b }   //: (Unit -> Ref[Int]^?{*}^)
\end{lstlisting}
Nevertheless, without names, they stop tracking the sharing of escaping
resources, but only that \emph{there are some resources}.
This suffices for tracking resources with scoped lifetimes.

Recent work \cite{DBLP:journals/pacmpl/XuBPO25} proposes refining top qualifiers using
existential types. This choice works for their setup, but appears unideal for tracking
reachability in more general settings. Illustrated below:

\begin{lstlisting}
// extRes: \E s <:~*. Container[Resource^s^]
let (a, resA) = extRes  // [..., a <:~*] |- resA ==> ??Container[Resource^?a^]^?resA^
let (b, resB) = extRes  // [..., b <:~*] |- resB ==> ??Container[Resource^?b^]^?resB^
\end{lstlisting}
If @extRes@ is unpacked multiple times, the resulting @a@ and @b@ still
represent the same reachability, but now have individual identities.
This complicates later uses of @resA@ and @resB@, where the type system
needs to reason about their sharing.
In contrast, using self-references, different \emph{unpacked} instances
of the same resource share the identity of the original resource:

\begin{lstlisting}
// extRes: \mus.Container[Resource^s^]^~*^
let resA = extRes  // [...] |- resA ==> ??Container[Resource^?extRes^]^?resA^
let resB = extRes  // [...] |- resB ==> ??Container[Resource^?extRes^]^?resB^
\end{lstlisting} %
\section{Related Work} \label{sec:related}

\paragraph{Tracking Reachability/Capturing in Types}
This work is closely related to the original development of reachability types
\cite{DBLP:journals/pacmpl/BaoWBJHR21} and its polymorphic variant \polylang
\cite{DBLP:journals/pacmpl/WeiBJBR24},
both of which have borrowed ideas from separation
logic~\cite{DBLP:conf/lics/Reynolds02,DBLP:conf/csl/OHearnRY01}.
\citet{DBLP:journals/pacmpl/BaoJ0BR25} then established a logical relation 
model for a monomorphic substrate of \polylang, and
\citet{DBLP:journals/pacmpl/0001HJBR25} extended \polylang's store model for
added expressiveness. Our development builds on the logical relation model \cite{DBLP:journals/pacmpl/BaoJ0BR25}, adapted to match the terminating substrate of \citet{DBLP:journals/pacmpl/0001HJBR25}.
Nevertheless, to address the shared limitation from prior work that their 
subtyping does not allow conversions involving self-reference,
our semantic model deviates from \citet{DBLP:journals/pacmpl/BaoJ0BR25} in 
nontrivial ways, by interpreting locations reachable from values
in a type-dependent manner instead of leaving them invariant.
Moreover, this work presents a mechanized algorithmic development, which has
no prior counterpart.

Closely related, capturing types \cite{DBLP:journals/toplas/BoruchGruszeckiOLLB23}
integrate escape checking and capability tracking into Scala 3,
drawing on modal type theory \cite{DBLP:journals/tocl/NanevskiPP08}
and using boxing to achieve capture tunneling with polymorphism.
Both capturing types and reachability types track named resources via variable names,
but differ in how they handle unnamed resources,
as discussed in \Cref{sec:alternative}.
Early capturing types use a top qualifier for all unnamed resources,
which suffices for escape detection but carries no concrete information
for separation reasoning.
\emph{Degrees of separation} \cite{DBLP:journals/pacmpl/XuBO24}
refines this with read/write distinctions at the cost of extra annotations,
while freshness markers in reachability types
\cite{DBLP:journals/pacmpl/WeiBJBR24}---and in this work---represent
unique resources directly but flag any sharing regardless of effect.
\emph{Reach capabilities} \cite{DBLP:journals/pacmpl/XuBPO25}
further refine top qualifiers with existentials to track internal resources,
sharing motivation with our self-references, but, as discussed in
\Cref{sec:alternative}, may complicate identity and separation reasoning.
To our knowledge,
no formal presentation has combined reach capabilities with separation checks;
our self-references continue to support them.

\paragraph{Ownership Types}

Originally developed for object-oriented programming, ownership type
systems~\cite{DBLP:conf/ecoop/NobleVP98,DBLP:conf/oopsla/ClarkePN98,DBLP:series/lncs/ClarkeOSW13,DBLP:conf/oopsla/PotaninNCB06}
control the access paths to objects and enforce heap invariants.
Many extensions have been developed on top of ownership types,
such as disjointness domains \cite{DBLP:conf/oopsla/BrandauerCW15} to express
local alias invariants and external uniqueness \cite{DBLP:conf/ecoop/ClarkeW03}
to relax the uniqueness restriction, among others.
While ownership types track sophisticated properties, they usually require
considerable annotation effort. %
As a remedy, ownership inference tools~\cite{DBLP:conf/ecoop/DietlEM11,
DBLP:conf/ecoop/HuangDME12} have been developed, combining points-to analysis,
constraint solving, and human interaction to achieve practical results.
As a prominent example, Rust~\cite{DBLP:conf/sigada/MatsakisK14,DBLP:journals/cacm/JungJKD21} implements a strong ownership model following the ``shared XOR mutable'' principle
and has seen wide adoption in systems programming.

Similarly, reachability types aim at regulating aliases in higher-order
languages.
Without a primitive notion of \emph{ownership}, they track sharing and separation
via reachability qualifiers.
While this suffices to achieve some access control,
patterns like uniqueness can only be realized by layering an additional effect system \cite{DBLP:journals/corr/abs-2510-08939}.
By separating effects and requiring an acyclic heap structure,
inference of reachability types is relatively straightforward,
as demonstrated by \algolang in this work.

\paragraph{The Avoidance Problem}
Systems with dependent types often encounter \emph{the avoidance problem}:
a term's type may mention a variable that has gone out of scope,
and no minimal supertype avoiding that variable exists.
The problem does not arise in systems with full dependent types,
where substitution is unrestricted; it appears in systems where variables
can only be substituted by other variables, such as
DOT \cite{DBLP:conf/oopsla/RompfA16},
bounded existentials \cite{DBLP:journals/tcs/GhelliP98},
and ML-style modules.
Consequently, some designs of ML-style module systems lack complete type
checking \cite{DBLP:journals/jfp/Leroy00};
refined designs avoid the problem by elaborating escaped variables
into implicit existential types
\cite{DBLP:conf/popl/DreyerCH03,DBLP:journals/jfp/RossbergRD14}.
In this work, variable escaping is an instance of the avoidance problem,
and our solution, the self-reference, plays a similar role to
implicit existential types in ML-style module systems.
However, self-references are not existential types, as they do not require
term-level packing/unpacking and can be smoothly introduced via subtyping.
Their subtyping also differs from that of bounded existential types
\cite{DBLP:journals/ngc/WehrT11}.

\paragraph{Self-References are not Recursive Types}
While we borrow the $\mu$-notation, self-references differ from recursive types
\cite{DBLP:journals/toplas/AmadioC93} in several respects.
Both let a construct refer to itself, but a self-reference refers to a
\emph{term}---specifically, a dependent variable naming the value being
typed---whereas recursive types refer to themselves at the type level.
In the @Counter@ example of \Cref{sec:intro}, the self-reference stands for
the counter value itself, and cannot be recursively expanded the way
a recursive type unfolds into a copy of itself.
Qualifiers are also finite sets:
self-references \emph{can} be unfolded by subtyping,
but such unfolding saturates, so we never face the infinite equivalent
or isomorphic qualifiers that make subtyping of recursive types subtle
\cite{DBLP:journals/pacmpl/ZhouZO23}.
Moreover, when a self-reference is unfolded, a function value of the desired
qualifier is guaranteed to be inhabited, avoiding the bad-bounds unsoundness
familiar from DOT-like systems \cite{DBLP:conf/oopsla/AminT16}.
We leave extending reachability types with genuine recursive types as future work.

\paragraph{Type Inference with Subtyping} %
The Hindley-Milner (HM) typing algorithm \cite{DBLP:conf/popl/DamasM82}
is able to infer principal types in polymorphic languages without subtyping.
HM(X) \cite{DBLP:journals/tapos/OderskySW99} extends HM to constrained types
including subtyping while preserving principal types.
More recent approaches aim for compactness, inferring more compact types
in the presence of subtyping \cite{DBLP:journals/pacmpl/Parreaux20,DBLP:conf/popl/DolanM17};
\citet{DBLP:journals/pacmpl/ParreauxC22} further integrates a set-algebraic
structure into the type system, which can be extended to track
non-escaping region variables \cite{DBLP:journals/pacmpl/GaoP25}.

As opposed to global inference algorithms, bidirectional typing
\cite{DBLP:journals/csur/DunfieldK21,DBLP:conf/popl/DunfieldP04}
features local reasoning at the cost of some explicit type annotations.
Our adaptation for $\maybelang$ layers qualifiers on top of bidirectional
typing: we infer qualifiers via an eager instantiation algorithm in the style of
\citet{DBLP:conf/icfp/DunfieldK13},
and employ a hybrid checking/synthesizing mode similar to the refinement
strengthening of \citet{DBLP:conf/pldi/PolikarpovaKS16}.
Type and qualifier \emph{instantiations} at polymorphic call sites, however,
remain explicit.
Inferring them is non-trivial: \algolang supports impredicative polymorphism,
for which implicit instantiation can render subtyping undecidable
\cite{DBLP:conf/mfcs/Chrzaszcz98} and lies beyond the \emph{complete-and-easy}
approach \cite{DBLP:conf/icfp/DunfieldK13}, calling for more
sophisticated constraint-based solutions
\cite{DBLP:journals/toplas/PierceT00,DBLP:conf/popl/OderskyZZ01,DBLP:conf/ecoop/ZhaoO22,DBLP:journals/pacmpl/CuiJO23};
moreover, the interleaving of types and qualifiers introduces additional
scoping and bound-computation concerns.
We leave this as future work.

\section{Conclusion} \label{sec:conc}

In this work, we investigated expressive subtyping with self-references and
inference of qualifiers in reachability types. We identified the limitations of
prior work and proposed a new declarative reachability type system $\maybelang$
with the combined notion of type-and-qualifier subtyping, enabling sound specification of
the subtyping behavior of self-references.  We also developed the first bidirectional
typing algorithm $\algolang$ for reachability types and proved that it is decidable and sound
with respect to $\maybelang$.  To evaluate the expressiveness
of our algorithm, we implemented and tested a variety of examples
involving data structures.
Both the declarative and algorithmic systems are mechanized in Lean.
Together, they advance the usability of reachability types, and we believe
they represent major steps towards practical programming languages. 
\section*{Data Availability Statement}

Our artifact is available on Zenodo \cite{artifact} and
\url{https://github.com/TiarkRompf/reachability/tree/main/checking/lean_v2}.

\begin{acks}
We thank Haotian Deng for related development on reachability types.
We thank Zhe Zhou, Patrick LaFontaine, Craig Liu, and Yueyang Tang for their
feedback on early drafts.
We thank anonymous reviewers for their valuable comments.
This work was supported in part by \grantsponsor{NSF}{NSF}{} award
\grantnum{NSF}{2348334} and an Augusta University faculty startup package,
as well as gifts from Meta, Google, Microsoft, and VMware.
\end{acks}

\newpage
\appendix
\renewcommand{\thefigure}{\thesection.\arabic{figure}}
\counterwithin{figure}{section}

\section{Additional Typing Rules} \label{app:typing}

\Cref{fig:closed} defines the well-formedness conditions for \maybelang,
and \Cref{fig:decpoly} presents the typing and subtyping rules related to
type polymorphism and type annotations in \maybelang.

\begin{figure}[ht]
\footnotesize\makeatother
\setlength{\afterruleskip}{\smallskipamount}

\judgement{Occurrence Predicate}{%
    \BOX{\strut x\notin T}
  \ \BOX{\strut x\notin^+ T}
  \ \BOX{\strut x\notin^- T}}\\[-1em]
\begin{tabular}{@{}p{.34\linewidth}@{}p{.33\linewidth}@{}p{.33\linewidth}@{}}
  \infax{
    \overline{ y \notin \TBase } \andalso
    \overline{ y \notin \TTop } \andalso
    \overline{ y \notin X } }
  &
  \infax{
    \overline{ y \notin^+ \TBase } \andalso
    \overline{ y \notin^+ \TTop } \andalso
    \overline{ y \notin^+ X } }
  &
  \infax{
    \overline{ y \notin^- \TBase } \andalso
    \overline{ y \notin^- \TTop } \andalso
    \overline{ y \notin^- X } }
  \\
  \infrule{
    y \notin T \quad y \notin p \qquad y \notin U \quad y \notin q
  }{
    y \notin \TSRef{h}{\ty[p]{T}}{\ty[q]{U}}
  }
  &
  \infrule{
    y \notin^- T \qquad y \notin^+ U \quad y \notin q
  }{
    y \notin^+ \TSRef{h}{\ty[p]{T}}{\ty[q]{U}}
  }
  &
  \infrule{
    y \notin^+ T \quad y \notin p \qquad y \notin^- U
  }{
    y \notin^- \TSRef{h}{\ty[p]{T}}{\ty[q]{U}}
  }
  \\
  \infrule{
    y \notin T \quad y \notin p \qquad y \notin U \quad y \notin q
  }{
    y \notin \TFun{f}{x}{\ty[p]{T}}{\ty[q]{U}}
  }
  &
  \infrule{
    y \notin^- T \qquad y \notin^+ U \quad y \notin q
  }{
    y \notin^+ \TFun{f}{x}{\ty[p]{T}}{\ty[q]{U}}
  }
  &
  \infrule{
    y \notin^+ T \quad y \notin p \qquad y \notin^- U
  }{
    y \notin^- \TFun{f}{x}{\ty[p]{T}}{\ty[q]{U}}
  }
  \\
  \infrule{
    y \notin T \quad y \notin p \qquad y \notin U \quad y \notin q
  }{
    y \notin \TAll{f}{X}{\ty[p]{T}}{\ty[q]{U}}
  }
  &
  \infrule{
    y \notin^- T \qquad y \notin^+ U \quad y \notin q
  }{
    y \notin^+ \TAll{f}{X}{\ty[p]{T}}{\ty[q]{U}}
  }
  &
  \infrule{
    y \notin^+ T \quad y \notin p \qquad y \notin^- U
  }{
    y \notin^- \TAll{f}{X}{\ty[p]{T}}{\ty[q]{U}}
  }
\end{tabular} \\[1em]

\judgement{Well-Formed Qualifiers, Types, and Qualified Types}{%
  \BOX{\strut \G\ts q} \ \BOX{\strut \G\ts T} \ \BOX{\strut \G\ts Q}}\\[-1em]
\typicallabel{c-fun}
\begin{tabular}{@{}p{.2\linewidth}@{}p{.24\linewidth}@{}p{.18\linewidth}@{}p{.18\linewidth}@{}p{.2\linewidth}@{}}
  \infrule[\ruledef{cqual}{c-qual}]{
    q \subs \qfresh{\dom\G}
  }{
    \G \ts q
  }
  &
  \infrule[\ruledef{cqtype}{c-qtype}]{
    \G\ts q \andalso \G\ts \ty{T}
  }{
    \G \ts \ty[q]{T}
  }
  &
  \infrule[c-base]{\quad}{\G \ts \TBase}
  &
  \infrule[c-top]{\quad}{\G \ts \TTop}
  &
  \infrule[c-tvar]{\ctxtyp{X}{Q} \in \G}{\G \ts X}
\end{tabular} \\
\begin{tabular}{@{}p{\linewidth}@{}}
  \infrule[c-ref]{
    \G,\ctxvar{h}{\TTop[\qfresh]}\ts \ty[p]{T} \qquad
    \G,\ctxvar{h}{\TTop[\qfresh]}\ts \ty[q]{U} \qquad
    \newcontent{h \notin^+ T} \qquad
    \newcontent{h \notin p} \qquad
    \newcontent{h \notin^- U}
  }{
    \G\ts \TSRef{h}{\ty[p]{T}}{\ty[q]{U}}
  }
  \\
  \infrule[\ruledef{cfun}{c-fun}]{
    \G,\ctxvar{f}{\TTop[\qfresh]}\ts \ty[p]{T} \qquad
    \G,\ctxvar{f}{\TTop[\qfresh]},\ctxvar{x}{\ty[p]{T}} \ts\ty[q]{U} \qquad
    \newcontent{f \notin^+ T} \qquad
    \newcontent{f \notin p} \qquad
    \newcontent{f \notin^- U}
  }{
    \G\ts \TFun{f}{x}{\ty[p]{T}}{\ty[q]{U}}
  }
  \\
  \infrule[c-all]{
    \G,\ctxvar{f}{\TTop[\qfresh]}\ts \ty[p]{T} \qquad
    \G,\ctxvar{f}{\TTop[\qfresh]},\ctxtyp{X}{\ty[p]{T}} \ts\ty[q]{U} \qquad
    \newcontent{f \notin^+ T} \qquad
    \newcontent{f \notin p} \qquad
    \newcontent{f \notin^- U}
  }{
    \G\ts \TAll{f}{X}{\ty[p]{T}}{\ty[q]{U}}
  }
\end{tabular} \\[1em]

\judgement{Well-Formed Context}{\BOX{\strut \WF{\G}}}\\[-1em]
\begin{tabular}{@{}p{.25\linewidth}@{}p{.3\linewidth}@{}p{.3\linewidth}@{}p{.15\linewidth}@{}}
  \infrule{\WF{\G} \qquad \G\ts\ty{Q}}{\WF{\G,\ \ctxvar{x}{Q}}}
  &
  \infrule{\WF{\G} \qquad \G\ts{\ty[q]{T}}}{\WF{\G,\ \newcontent{\ctxvar{f}{\ty[q,\qhole]{T}}}}}
  &
  \infrule{\WF{\G} \qquad \G\ts\ty{Q}}{\WF{\G,\ \ctxtyp{X}{Q}}}
  &
  \infrule{\quad}{\WF{\qbot}}
\end{tabular}\\[-1ex]
\caption{Selected rules of well-formedness in $\maybelang$.
Compared with \cite{DBLP:journals/pacmpl/WeiBJBR24} (\newcontent{\text{shaded}}\!),
qualifier holes are permitted in the typing context,
and self-references are permitted in deep, covariant positions.} \label{fig:closed}
\vspace{-3ex}
\end{figure} 
\begin{figure}[ht]
\footnotesize \typicallabel{t-get} \makeatother
\setlength{\afterruleskip}{\smallskipamount}
\judgement{Additional Typing}{\BOX{\G[\flt] \ts t : Q}} \\[-1em]
\begin{tabular}{@{}p{.28\linewidth}@{}p{.32\linewidth}@{}p{.4\linewidth}@{}}
  \infrule[\ruledef{tanno}{t-as}]{
    \phantom{text} \\
    \G[\flt]\ts t : Q
  }{
    \G[\flt]\ts \tast{t}{Q} : Q
  }
  &
  \infrule[\ruledef{tabsa}{t-abs'}]{
    F = \TFun{f}{x}{P}{Q} \\
    \G[\flt] \ts \tlam{f}{x}{t} : \ty[q]{F}
  }{
    \G[\flt] \ts \tlam[P]{f}{x}{t} : \ty[q]{F}
  }
  &
  \infrule[\ruledef{ttabsa}{t-tabs'}]{
    F = \TAll{f}{X}{P}{Q} \\
    \G[\flt] \ts \ttlam{f}{X}{t} : \ty[q]{F}
  }{
    \G[\flt] \ts \ttlam[P]{f}{X}{t} : \ty[q]{F}
  }
\end{tabular} \\[-1ex]
\caption{Additional typing rules of \maybelang,
related to type annotations.} \label{fig:decpoly}
\end{figure}
 
\subsection{Mechanized Extension for Arbitrary Argument Reachability}
\label{sec:arbitraryarg}

In addition to the two conformance conditions defined in \Cref{fig:typing},
we mechanize an additional mode that accepts arbitrary parameter reachability,
if the function argument qualifier includes both the freshness marker $\qfresh$
and the self-reference $f$ of the current function:

{
\footnotesize \setlength{\afterruleskip}{\medskipamount}
\infrule[f-wild]{
    f~\text{is the self-reference of the current function}
}{
    \G[\flt] \ts s \subapp{q} \qfresh{f},p
}
}

To enable this wildcard case, we generalize the definition of well-formed
function types (\Cref{fig:closed}) to allow self-references in the argument
qualifier \emph{only if} they are accompanied by freshness markers.
Semantically, since self-references denote the observation of the function,
and freshness denotes resources not observed in the function, their combination
does not cause a variance issue, but means arbitrary resources.
Subtyping rules can also be generalized, so that when the subtype function
includes $\qfresh{f}$ in the argument qualifier,
the supertype argument qualifier can be arbitrary.

Crucially, this interpretation does not yield a general \emph{top} qualifier,
but can only be used in function arguments, and requires referring to the
most recent self-reference $f$.
Precisely describing this behavior in rules is subtle and thus omitted
from the main text, but both \maybelang and \algolang have supported this
extension in typing and subtyping, which also comes in handy in spelling
examples. %
\section{Semantic Soundness of \maybelang} \label{sec:lr}

\newcommand{\VV}{\ensuremath{\mathbb{V}}}
\newcommand{\VVV}{\ensuremath{\mathcal{V}}}
\newcommand{\interp}[2][]{\ensuremath{[\![#2]\!]\ifstrempty{#1}{}{_{#1}}}}
\newcommand{\mkTpl}[1]{\ensuremath{\langle#1\rangle}}
\newcommand{\tforall}[1]{\ensuremath{\forall #1.\ }}
\newcommand{\texists}[1]{\ensuremath{\exists #1.\ }}
\newcommand{\timplies}{\ensuremath{\Rightarrow\ }}
\newcommand{\stsub}[1][]{\ensuremath{\sqsubseteq\ifstrempty{#1}{}{_{#1}}}}
\newcommand{\valintp}[1]{\ensuremath{V\!\interp{#1}}}
\newcommand{\envintp}[2]{\ensuremath{G_{#1}\!\interp{#2}}}
\newcommand{\expintp}[2]{\ensuremath{E\interp[#2]{#1}}}
\newcommand{\sttyp}{\ensuremath{\!:\!}}
\newcommand{\steff}[1]{\ensuremath{\!\to_{#1}\!}}
\newcommand{\tinv}[1]{\ensuremath{\mathsf{Inv}_{#1}}}
\newcommand{\defeq}{\ensuremath{\overset{\mathrm{def}}{=}}}

Unlike prior work~%
\cite{DBLP:journals/pacmpl/BaoWBJHR21,DBLP:journals/pacmpl/WeiBJBR24,DBLP:journals/pacmpl/0001HJBR25}
that establishes syntactic soundness,
we prove the type soundness of $\maybelang$ using \emph{logical relations}~%
\cite{DBLP:journals/jacm/TimanyKDB24}. The dynamic semantics is formulated as a
big-step interpreter~\cite{DBLP:conf/popl/AminR17}.
Our semantic interpretation is adapted from~\citet{DBLP:journals/pacmpl/BaoJ0BR25},
and we extend their model by adding interpretations for type polymorphism,
shallowly qualified, dual-component references~\cite{DBLP:journals/pacmpl/0001HJBR25},
and deep occurrences of bound variables inside types.
Besides efforts to reflect recent advances in reachability types,
our new subtyping design requires our semantic model to interpret the
reachable locations of values in a type-dependent manner,
and this shift necessitated nontrivial changes.
Here, we provide details for our logical relation model.
Keys definitions are presented in \Cref{fig:lrdefs}. We follow the standard
erasure semantics for lambda calculus with references in the big-step style
($\Downarrow$), and we omit their definitions here.

Compared with \citet{DBLP:journals/pacmpl/BaoJ0BR25}, a key difference
lies in how we denote qualifiers to locations.
While they require computing reachable locations from the values of each
involved variable using the value environment ($H$ in their presentation, $\rho$
in ours), in our logical relation setup, we denote qualifiers using a dedicated
\emph{location environment} $H$.
While $H$ and $\rho$ are still related by environment interpretation in our
formalization, locations are no longer strictly determined by values.

\begin{figure}[ht]
\makeatother \footnotesize \renewcommand{\arraystretch}{1.2}
\[\begin{array}{lcll}
  v & ::= & c \mid \ell \mid \mkTpl{\rho, \tlam{f}{x}{t}}
           \mid \mkTpl{\rho, \ttlam{f}{X}{t}} & \text{Values} \\
  L &\in & \mathcal{P}_{\mathsf{fin}}(\mathsf{Locs}) & \text{Locations} \\
  \VV    & ::= & \{\mkTpl{v,L}\} \mid \{\mkTpl{\Sigma,v,L}\} & \text{Value Type} \\[1ex]
  \sigma & ::= & \qbot \mid \sigma,\ctxvar{\ell}{v} & \text{Store} \\
  \Sigma & ::= & \qbot \mid \Sigma,\ctxvar{\ell}{\mkTpl{\VV,L}} & \text{Store Typing} \\
  \rho   & ::= & \qbot \mid \rho,\ctxvar{x}{v} & \text{Value Environment} \\
  H      & ::= & \qbot \mid H,\ctxvar{x}{L} & \text{Location Environment} \\
  \VVV   & ::= & \qbot \mid \VVV,\ctxvar{X}{\VV} & \text{Value Type Environment} \\[1ex]
  \interp[H]{q} & = & \bigcup_{x\in q} H(x) & \text{Qualifier Denotation} \\
\end{array}
\]
\vspace{-1ex}
\caption{Definitions in logical relations.} \label{fig:lrdefs}
\end{figure}

\paragraph{Store Predicates}

In \Cref{fig:storepreds}, we define the store predicates,
including well-typed stores, store typing extension, and store effects.
These definitions relate stores to store typings, and describe the invariants
when they change.

\begin{figure}[t]
\makeatother \footnotesize \renewcommand{\arraystretch}{1.2}
\[\begin{array}{lclll}
  \sigma\sttyp\Sigma & \defeq &
    \dom{\sigma} = \dom{\Sigma}
      \land (\tforall{\ell\in \dom{\sigma}} & \text{Well-Typed Store} \\ & & \quad
    \texists{\VV\,L} \Sigma(\ell) = \mkTpl{\VV, L} \land
      \texists{L' \subs L} \mkTpl{\sigma(\ell), L'} \in \VV) \\

  \interp[\Sigma]{\ell} & \defeq & \{\ell\} \cup \interp[\Sigma]{L} \qquad
    \text{where}~\Sigma(\ell) = \mkTpl{\VV, L} & \text{Location Transitive Closure} \\
  \interp[\Sigma]{L} & \defeq & \bigcup_{\ell\in L} \interp[\Sigma]{\ell} \\

  \Sigma \stsub[L] \Sigma' & \defeq &
    L\subs\dom{\Sigma} \land L\subs\dom{\Sigma'} \land
    (\tforall{\ell \in L} \Sigma(\ell) = \Sigma'(\ell)) & \text{Store Typing Extension} \\
  \Sigma \stsub[L]^* \Sigma' & \defeq &
    \Sigma \stsub[{\interp[\Sigma]{L}}] \Sigma' \\
  \Sigma \stsub \Sigma' & \defeq &
    \Sigma \stsub[\dom{\Sigma}] \Sigma' \\

  \sigma\steff{L}\sigma' & \defeq &
    \tforall{\ell \in \dom{\sigma}} \ell \notin L \timplies
    \sigma(\ell) = \sigma'(\ell) & \text{Store Effects} \\
\end{array}
\]
\vspace{-1ex}
\caption{Store predicates in logical relations.} \label{fig:storepreds}
\end{figure}

\paragraph{Value Interpretation}

In \Cref{fig:vtp}, we present the definitions of our value interpretations, or
value typing. Such interpretation is indexed by types, and summarizes all
possible values of the type.
Compared with \citet{DBLP:journals/pacmpl/BaoJ0BR25}, our value interpretation
involves the location environments $H$ instead of the value environments $\rho$
(written $H$ in their formalization); we add the value type environment $\VVV$
to support type polymorphism; and we include locations $L$ in value
interpretations.
Thus, locations a value may reach can vary.
This is the key to supporting variance of self-reference locations in our
semantic subtyping.

\begin{theorem}[Value Locations May Grow]
If $\mkTpl{\Sigma,H,\VVV,v,L} \in \valintp{T}$,
and $L\subs L'\subs \dom{\Sigma}$,
there is $\mkTpl{\Sigma,H,\VVV,v,L'} \in \valintp{T}$.
\end{theorem}

\begin{figure}[ht]
\footnotesize \renewcommand{\arraystretch}{1.2}
\judgement{Value Interpretation}{}
\[\begin{array}{rll}
  \valintp{\TBase} & = &
    \{\mkTpl{\Sigma, H, \VVV, c, L} \mid L \subs\dom{\Sigma}\} \\

  \valintp{\TSRef{h}{\ty[p]{T}}{\ty[q]{U}}} & = &
    \{\mkTpl{\Sigma, H, \VVV, \ell, L} \mid
      \ell\in L\subs\dom{\Sigma} \land
      \texists{\VV'\, L'}
      \Sigma(\ell) = \mkTpl{\VV', L'}\ \land
      \\ & & \hspace{1em}
      (\tforall{\sigma'\sttyp\Sigma'}
        \Sigma \stsub[L]^* \Sigma' \timplies
        \tforall{H'\! = H,\ctxvar{h}{L}}
        \\ & & \hspace{2em}
        (\tforall{v} \tforall{L_1 \subs \interp[H]{p}}
          \mkTpl{\Sigma', H', \VVV, v, L_1} \in \valintp{T} \timplies
          \texists{L_2\subs L'}
          \mkTpl{v, L_2} \in \VV')\ \land
        \\ & & \hspace{2em}
        (\tforall{v} \tforall{L_1 \subs L'}
          \mkTpl{v, L_1} \in \VV' \timplies
          \texists{L_2\subs \interp[H']{q}}
          \mkTpl{\Sigma', H', \VVV, v, L_2} \in \valintp{U}))\} \\

  \valintp{ \TFun{f}{x}{\ty[p]{T}}{\ty[q]{U}}} & = &
    \{\mkTpl{\Sigma, H, \VVV, \mkTpl{\rho, \tlam{f}{x}{t}}, L} \mid
      L \subs \dom{\Sigma}\ \land
      \\ & & \hspace{1em}
      (\tforall{\sigma_1\sttyp\Sigma_1}
        \Sigma \stsub[L]^* \Sigma_1 \timplies
        \tforall{H_1\! = H,\ctxvar{f}{L}}
        \\ & & \hspace{2em}
        \tforall{v_1} \tforall{L_1 \subs \interp[H_1]{\qclean{p}}
          \cup_{\qfresh\in p} \overline{L}}
          \mkTpl{\Sigma_1, H_1, \VVV, v_1, L_1} \in \valintp{T} \timplies
        \\ & & \hspace{2em}
        \tforall{\rho'\! = \rho,\ctxvar{f}{\mkTpl{\rho,\tlam{f}{x}{t}}},\ctxvar{x}{v_1}}
          \tforall{H_2\! = H_1,\ctxvar{x}{L_1}}
          \\ & & \hspace{3em}
          \texists{v_2} \texists{\sigma_2\sttyp\Sigma_2}
            {\sigma_1,\rho',t}\Downarrow{\sigma_2, v_2} \land
            \Sigma_1 \stsub \Sigma_2 \land
            \sigma_1\steff{L,L_1}\sigma_2 \ \land
          \\ & & \hspace{3em}
          \texists{L_2 \subs \interp[H_2]{\qclean{q}}
              \cup_{\qfresh\in q} \overline{\dom{\Sigma_1}}}
            \mkTpl{\Sigma_2, H_2, \VVV, v_2, L_2} \in \valintp{U})\} \\

  \valintp{\TTop}  & = &
    \{\mkTpl{\Sigma, H, \VVV, v, L} \mid
      L \subs\dom{\Sigma}\} \\

  \valintp{X} & = &
    \{\mkTpl{\Sigma, H, \VVV, v, L} \mid
      L \subs \dom{\Sigma} \land
      (\tforall{\Sigma'} \Sigma \stsub[L]^* \Sigma' \timplies \\ & & \qquad
        \tforall{L' \subs \dom{\Sigma'}} L\subs L' \timplies
        \mkTpl{\Sigma', v, L'} \in \VVV(X)\ )\} \\

  \valintp{ \TAll{f}{x}{\ty[p]{T}}{\ty[q]{U}}} & = &
    \{\mkTpl{\Sigma, H, \VVV, \mkTpl{\rho, \ttlam{f}{x}{t}}, L} \mid
      L \subs \dom{\Sigma}\ \land \\ & & \qquad
      (\tforall{\sigma_1\sttyp\Sigma_1}
        \Sigma \stsub[L]^* \Sigma_1 \timplies
        \tforall{H_1\! = H,\ctxvar{f}{L}}\\ & & \qquad
      \ \tforall{\VV} (\tforall{\sigma'\sttyp\Sigma'} \tforall{v'\,L'}
        \mkTpl{\Sigma', v', L'}\in\VV \timplies
        \mkTpl{\Sigma', H_1,\VVV, v', L'}\in\valintp{T}) \timplies \\ & & \qquad
      \ \tforall{L_1 \subs \dom{\Sigma_1}}
        L_1 \subs \interp[H_1]{\qclean{p}} \cup_{\qfresh\in p} \overline{L}
        \timplies \\ & & \qquad
      \ \tforall{\rho'\! = \rho,\ctxvar{f}{\mkTpl{\rho,\ttlam{f}{x}{t}}},\ctxvar{x}{c}}
        \tforall{H_2\! = H_1,\ctxvar{x}{L_1}}
        \tforall{\VVV_2\! = \VVV,\ctxvar{X}{\VV}}\\ & & \qquad\qquad
        \texists{v_2} \texists{\sigma_2\sttyp\Sigma_2}
          {\sigma_1,\rho',t}\Downarrow{\sigma_2, v_2} \land
          \Sigma_1 \stsub \Sigma_2 \land
          \sigma_1\steff{L,L_1}\sigma_2\ \land \\ & & \qquad\qquad
        \texists{L_2 \subs \interp[H_2]{\qclean{q}}
            \cup_{\qfresh\in q}\overline{\dom{\Sigma_1}}}
          \mkTpl{\Sigma_2, H_2,\VVV_2, v_2, L_2} \in \valintp{U})\} \\
\end{array}\]

\vspace{-1ex}
\caption{Definitions for value interpretations.} \label{fig:vtp}
\end{figure}

\paragraph{Environment Interpretation}

With value interpretations defined, we then define environment interpretations
in \Cref{fig:lrelse}, relating different sorts of typing environments,
specifically values $\rho$, locations $H$, and value types $\VVV$, with respect
to the observation filter $\flt$.
Whereas \citet{DBLP:journals/pacmpl/BaoJ0BR25} involves a single environment
interpretation with both $\tinv{o}$ (for individual entries) and $\tinv{s}$
(for location separation), we separate it into two, as our semantic subtyping
is defined without $\tinv{s}$.
Additionally, we add additional $\tinv{o}$ cases for describing self-references
and type-and-qualifier bound entries.
Of particular interest to this work, while qualifier denotations are upper 
bounds to variable locations, they are lower bounds to self-reference locations
(modulo freshness), justifying the flipped directions between \rulename{qself}
and \rulename{qvar}.

\begin{figure}[ht]
\makeatother \footnotesize \renewcommand{\arraystretch}{1.2}
\judgement{Environment Interpretation}{}
\[\begin{array}{l@{\ }l@{\quad}l}

  \tinv{o}(\Sigma,\rho,H,\VVV,\flt) & := &
    \tforall{\ctxvar{x}{\ty[q]{T}} \in \G}
      (x \in \flt \timplies
        \mkTpl{\Sigma, H, \VVV, \rho(x), H(x)} \in \valintp{T}) \land
      (\qfresh\notin q \timplies
        H(x) \subs \interp[H]{q}) \\

  \tinv{o}^{\mu}(\Sigma,\rho,H,\VVV,\flt) & := &
    \tforall{\ctxvar{f}{\ty[q]{T}} \in \G}
      (f \in \flt \timplies
        \mkTpl{\Sigma, H, \VVV, \rho(f), H(f)} \in \valintp{T}) \land
      (\qfresh\notin q \timplies H(f) \subs \interp[H]{q})\ \land\\ & & \qquad
      \interp[H]{\qclean{q}} \subs H(f) \\

  \tinv{o}^{<:}(\Sigma,\rho,H,\VVV,\flt) & := &
    \tforall{\ctxtyp{X}{\ty[q]{T}} \in \G}
      (x \in \flt \timplies H(x) \subs \dom{\Sigma}) \land
      (\qfresh\notin q \timplies H(x) \subs \interp[H]{q})\ \land\\ & & \qquad
      (\tforall{\sigma'\sttyp\Sigma'}\tforall{v'\,L'}
        \mkTpl{\Sigma',v',L'} \in \VVV(X) \timplies
        \mkTpl{\Sigma',H,\VVV,v',L'} \in \valintp{T}) \\

  \envintp{1}{\G[\flt]} & := & \{\mkTpl{\Sigma,\rho,H,\VVV} \mid
    \flt \subs \dom{\G} = \dom{\rho} = \dom{H}\ \land \\ & & \qquad
    \tinv{o}(\Sigma,\rho,H,\VVV,\flt) \land
    \tinv{o}^{\mu}(\Sigma,\rho,H,\VVV,\flt) \land
    \tinv{o}^{<:}(\Sigma,\rho,H,\VVV,\flt) \}\\

  \tinv{s}(H, \flt) & := &
    \tforall{p\subs \qfresh{\flt}} \tforall{q\subs \qfresh \flt}
      p \overlap q \subs \qfresh\flt \timplies
      \interp[H]{\qclean{p}} \cap \interp[H]{\qclean{q}} \subs
        \interp[H]{\qclean{(p \overlap q)}} \\

  \envintp{2}{\G[\flt]} & := & \{\mkTpl{\Sigma,\rho,H,\VVV} \mid
    \mkTpl{\Sigma,\rho,H,\VVV} \in \envintp{1}{\G[\flt]} \land
    \tinv{s}(H, \flt) \} \\
\end{array}\]

\judgement{Semantic Subtyping}{}
\[\begin{array}{lll}
  \G \mts p \subq q & \defeq &
    \tforall{\flt} \tforall{\mkTpl{\Sigma,\rho,H,\VVV} \in
      \envintp{1}{\G[\flt]}}
    \qfresh \notin p \setminus q \land
    \interp[H]{\qclean{p}} \subs \interp[H]{\qclean{q}} \\
  
  \G \mts \ty[p]{T} \subt \ty[q]{U} & \defeq &
    \tforall{\flt} \tforall{\mkTpl{\Sigma,\rho,H,\VVV} \in
        \envintp{1}{\G[\flt]}}
      \qfresh \notin p \setminus q \land
        \interp[H]{\qclean{p}} \subs \interp[H]{\qclean{q}}\ \land
      \\ & & \hspace{1em}
      (\tforall{\sigma\sttyp\Sigma}
        \interp[H]{\qclean{q}}\subs\interp[H]{\flt}\timplies
        \\ & & \hspace{2em}
        \tforall{v\,L}
          (\qfresh\notin p \timplies L \subs \interp[H]{p}) \timplies
          \mkTpl{\Sigma,H,\VVV,v,L} \in \valintp{T} \timplies
        \\ & & \hspace{2em}
        \texists{L' \subs \interp[H]{\qclean{q}}}
          \mkTpl{\Sigma,H,\VVV,v,L\cup L'} \in \valintp{U})
\end{array}\]

\judgement{Expression Interpretation and Semantic Typing}{}
\[\begin{array}{lcl}
  \expintp{\ty[q]{T}}{\flt} & := & \{ \mkTpl{\sigma,\Sigma,\rho,H,\VVV,t} \mid
    \texists{\sigma'\sttyp\Sigma'} \texists{v}
    {\sigma,\rho,t}\Downarrow{\sigma',v} \land
    \Sigma\stsub\Sigma' \land
    \sigma\steff{\interp[H]{\flt}}\sigma' \ \land \\ & & \qquad
      \texists{L \subs \interp[H]{q\cap\flt}
        \cup_{\qfresh\in q}\overline{\dom{\Sigma}}}
      \mkTpl{\Sigma',H,\VVV,v,L} \in \valintp{T} \} \\
  
  \G[\flt] \mts t: \ty[q]{T} & \defeq &
    \tforall{\mkTpl{\Sigma,\rho,H,\VVV} \in \envintp{2}{\G[\flt]}}
    \tforall{\sigma\sttyp\Sigma}
      \mkTpl{\sigma,\Sigma,\rho,H,\VVV,t} \in \expintp{\ty[q]{T}}{\flt}
\end{array}\]

\vspace{-1ex}
\caption{Environment and expression interpretations, and semantic typing and subtyping.}
\label{fig:lrelse}
\end{figure}

\paragraph{Semantic Subtyping and Typing}

With all definitions in place, we define semantic subtyping and typing
in \Cref{fig:lrelse}. The typing definition is similar to that defined in
\citet{DBLP:journals/pacmpl/BaoJ0BR25}, and we describe the subsumption
relations here.
The subqualifying relation $p <: q$ means that the denotation of $p$ is no
greater than $q$ for all permissible contexts.
The subtyping relation states that all values that can be interpreted as $T$
can also be interpreted as $U$, potentially with additional locations $L'$ no
greater than the denotation of qualifier $q$.
This allows self-references to observe changes in qualifiers during subtyping.

\begin{theorem}[Fundamental]
If a term $t$ is syntactically well-typed, \ie, \:$\G[\flt]\ts t: \ty[q]{T}$, and both the context $\G$ and the store $\sigma$ are well-formed,
then $t$ is also semantically well-typed (\:$\G \models t : \ty[q]{T} $).
Specifically, $t$ evaluates to a value $v$ of the type $\ty{T}$ in finite steps,
such that $v$ may only reach locations described by the qualifier $q$,
and all store write effects are limited to the locations described by $\flt$.
\end{theorem}

In particular, for terms closed under the empty context, we obtain a formulation
of type safety that requires no assumption about contexts or stores:
\emph{well-typed terms do not get stuck}.

\begin{corollary}[Type Safety]
If $\:\varnothing \ts t : \ty[\varnothing]{T}$, then $t$ evaluates to a value $v$
of type $T$ in finitely many steps.
\end{corollary}

The \emph{Preservation of Separation} property
\cite{DBLP:journals/pacmpl/WeiBJBR24} also follows from our analysis of store
effects: the evaluation of two well-typed terms with disjoint observations will
observe and update disjoint portions of the store. %
\section{Additional Algorithm Definitions and Metatheory} \label{app:algo}

We present the additional subtyping rules for type polymorphism in \Cref{fig:subcheckmore},
the definitions of polarized substitution in \Cref{fig:avoidmore},
and additional definitions for bidirectional typing in \Cref{fig:bidir-suppl}.
We detail the metatheory of \algolang as follows.

\begin{figure}[t]
\footnotesize\makeatother \typicallabel{av$_1$-no}
\setlength{\afterruleskip}{\smallskipamount}

\judgement{Subtype Checking (Continue)}{\quad} \\[-1em]
\begin{tabular}{@{}p{\linewidth}@{}}
  \infrule[\ruledef{suall}{su$_3$-all}]{
    \qfresh{} \notin q \andalso \theta = [q/f]
  }{
    \TAll{f}{X}{\ty[p]{T}}{Q} \,\subt[1]^{q}\,
    \algout{\TAll{f}{X}{\ty[p]{T\!\theta}}{Q\!\theta}}
  }
\end{tabular} \\[-1ex]
\caption{Additional subtype checking rules for type polymorphism.
Outputs are marked in \algout{red}.}
\label{fig:subcheckmore}
\end{figure}

\begin{figure}[t]
\footnotesize\makeatother \typicallabel{av$_1$-no}
\setlength{\afterruleskip}{\smallskipamount}

\judgement{Polarized Substitution}{
    \BOX{\strut\algout{T} := T [q/^+x]}
  \ \BOX{\strut\algout{T} := T [q/^-x]}}\\[-1em]
\begin{tabular}{@{}p{.5\linewidth}@{}p{.5\linewidth}@{}}
  \infrule[\ruledef{posref}{ps$^+$-ref}]{
    \algout{T_1} := T [q/^-y] \andalso
    {p_1} = p \setminus y \\
    \algout{U_1} := U [q/^+y] \andalso
    {r_1} = r [q/y] \\
    {V} = \TSRef{h}{\ty[p_1]{T_1}}{\ty[r_1]{U_1}}
  }{
    \algout{V} := \TSRef{h}{\ty[p]{T}}{\ty[r]{U}} [q/^+y]
  }
  &
  \infrule[\ruledef{negref}{ps$^-$-ref}]{
    \algout{T_1} := T [q/^+y] \andalso
    {p_1} = p [q/y] \\
    \algout{U_1} := U [q/^-y] \andalso
    {r_1} = r \setminus y \\
    {V} = \TSRef{h}{\ty[p_1]{T_1}}{\ty[r_1]{U_1}}
  }{
    \algout{V} := \TSRef{h}{\ty[p]{T}}{\ty[r]{U}} [q/^-y]
  }
  \\
  \infrule[\ruledef{posfun}{ps$^+$-fun}]{
    \algout{T_1} := T [q/^-y] \andalso
    {p_1} = p \setminus y \\
    \algout{U_1} := U [q/^+y] \andalso
    {r_1} = r [q/y] \\
    {V} = \TFun{f}{x}{\ty[p_1]{T_1}}{\ty[r_1]{U_1}}
  }{
    \algout{V} := \TFun{f}{x}{\ty[p]{T}}{\ty[r]{U}} [q/^+y]
  }
  &
  \infrule[\ruledef{negfun}{ps$^-$-fun}]{
    \algout{T_1} := T [q/^+y] \andalso
    {p_1} = p [q/y] \\
    \algout{U_1} := U [q/^-y] \andalso
    {r_1} = r \setminus y \\
    {V} = \TFun{f}{x}{\ty[p_1]{T_1}}{\ty[r_1]{U_1}}
  }{
    \algout{V} := \TFun{f}{x}{\ty[p]{T}}{\ty[r]{U}} [q/^-y]
  }
  \\
  \infrule[\ruledef{posall}{ps$^+$-all}]{
    \algout{T_1} := T [q/^-y] \andalso
    {p_1} = p \setminus y \\
    \algout{U_1} := U [q/^+y] \andalso
    {r_1} = r [q/y] \\
    {V} = \TAll{f}{X}{\ty[p_1]{T_1}}{\ty[r_1]{U_1}}
  }{
    \algout{V} := \TAll{f}{X}{\ty[p]{T}}{\ty[r]{U}} [q/^+y]
  }
  &
  \infrule[\ruledef{negall}{ps$^-$-all}]{
    \algout{T_1} := T [q/^+y] \andalso
    {p_1} = p [q/y] \\
    \algout{U_1} := U [q/^-y] \andalso
    {r_1} = r \setminus y \\
    {V} = \TAll{f}{X}{\ty[p_1]{T_1}}{\ty[r_1]{U_1}}
  }{
    \algout{V} := \TAll{f}{X}{\ty[p]{T}}{\ty[r]{U}} [q/^-y]
  }
\end{tabular} \\[-1ex]
\caption{Polarized substitution in \algolang.
Outputs are marked in \algout{red}.}
\label{fig:avoidmore}
\vspace{-3ex}
\end{figure} %
\begin{figure}[tp]
\footnotesize\makeatother \typicallabel{tiapp}
\setlength{\afterruleskip}{\medskipamount}

\judgement{Bidirectional Typing (Continue)}{
    \BOX{\strut\G[\algout{\flt}] \ts t \synth \algout{Q} \ots{\G}}}\\[-1em]
\begin{tabular}{@{}p{\linewidth}@{}}
  \infrule[\ruledef{tctabs}{tc$_2$-tabs}]{
    \pcx[\algout{\flt}]{\G, \ctxvar{f}{\TTop[\qhole]}, \ctxtyp{X}{\ty[p]{T}}}
      \ts t \chek \ty[q]{U} \ots{\G',\ctxvar{f}{\TTop[q,\qhole]},..} \\
    r = (p,q,\flt)\setminus\{\qfresh{f,x}\}
  }{
    \G[\algout{r}] \ts \ttlam{f}{X}{t} \chek \TAll[\synth\algout{r}]{f}{x}{\ty[p]{T}}{\ty[q]{U}} \ots{\G'}
  }
  \\
  \infrule[\ruledef{titabs}{ti$_7$-tabs}]{
    \G,\ctxvar{f}{\TTop[\qfresh]} \ts \ty[p]{T}
    \andalso f \notin^+ T \andalso f \notin p \\
    \pcx[\algout{\flt}]{\G, \ctxvar{f}{\TTop[\qhole]}, \ctxtyp{X}{\ty[p]{T}}}
      \ts t \synth \algout{\ty[q]{U}} \ots{\G', \ctxvar{f}{\TTop[q,\qhole]},..} \\
    \algout{V}\!:=\!U[f/^+f] \andalso
    {r} = (p,q,\flt)\setminus\{\qfresh{f,x}\} %
  }{
    \G[\algout{r}] \ts \ttlam[{\ty[p]{T}}]{f}{X}{t} \synth \algout{\TAll[r]{f}{X}{\ty[p]{T}}{\ty[q]{V}}} \ots{\G'}
  }
  \\
  \infrule[\ruledef{titapp}{ti$_{10}$-tapp}]{
    \G[\algout{\flt_1}] \ts t \synth\qexp
      \algout{\TAll[q]{f}{X}{\ty[p]{T}}{Q}} \ots{\G_1} \\
    \qfresh \notin q \lor f \notin T \andalso
    \G_1 \ts \ty[\qfresh]{V} \subt[] \ty[\algout{\qbot}]{T[q/f]} \ots{\G_2} \\
    \G[\algout{\flt_2}]_2 \ts s \subapp[]{q} p \ots{\G_3} \\
    Q \ll_{s/x}\ll_{q/f} \algout{\ty[r]{U}} \phantom{mn}
    \flt = \flt_1,\flt_2,(\qclean{s}),(r\!\setminus\!\{\qfresh{f,x}\})
  }{
    \G[\algout{\flt}] \ts \tapp{t}{\ty[s]{V}} \synth \algout{\ty[r]{U} [\ty[s]{V}/\ty[x]{X},q/f]} \ots{\G_3}
  }
\end{tabular} \\[1em]

\judgement{Algorithmic Application Conformance}{
  \BOX{\G[\algout{\flt}] \ts q \subapp[]{q} q \ots{\G}}} \\[-1em]
\begin{tabular}{@{}p{.4\linewidth}@{}p{.6\linewidth}@{}}
  \infrule[\ruledef{fasub}{fa-sub}]{
    \G \ts s \subq[] p \ots{\G'}
  }{
    \G[\algout{\qbot}] \ts s \subapp[]{q} p \ots{\G'}
  }
  &
  \infrule[\ruledef{fafresh}{fa-fresh}]{
    \G \ts s \overlap q \subq[] \qfresh{p} \ots{\G'} \andalso
    \qhole{}\notin \qsat{s},\qsat{q} \andalso
    {\flt} = \qclean{(s \overlap q)}
  }{
    \G[\algout{\flt}] \ts s \subapp[]{q} \qfresh{p} \ots{\G'}
  }
\end{tabular} \\[-1ex]
\caption{Rules for type abstraction and application in bidirectional typing in \algolang, along with definitions for type exposure and algorithmic application
conformance.
} \label{fig:bidir-suppl}
\vspace{-1em}
\end{figure} 
At a high level, all procedures of \algolang terminate, and they are sound with
respect to the declarative calculus~\maybelang (\Cref{thm:algtop}).
By necessity, \algolang is incomplete.
Since we implement the kernel variant of $F_{<:}$, we cannot accept terms valid
only in the full variant.
Besides, it is shown that qualifier
unification and thus inference have no principal solution (\Cref{thm:nonprincipal}).
Still, \algolang
produces reasonable results: qualifier exposure and thus \rulename{qcheck} are
proved complete (\Cref{thm:qtp_complete}), self unpacking ($\subt[1]$) results
are proved minimal (\Cref{lemma:self-unpack-equiv}), and the
avoidance strategy using the outermost self-reference is better than other
choices (\Cref{thm:minavoid}).
Given the nontrivial nature of sharing and separation reasoning, as
a type-based approach, we prefer termination over completeness.
In \Cref{sec:eval},
we show that \algolang is expressive enough to check meaningful programs.

\subsection{Qualifier Checking and Inference}
\label{sec:algmeta2}

We present the soundness and decidability of qualifier checking and inference, then show the completeness of qualifier checking, and explain why completeness is impossible for inference.
We start by analyzing qualifier checking \rulename{qcheck}.

\begin{theorem}[Decidability and Soundness of Qualifier Checking]
  If $\ \WF{\G}$, then it can be decided in finite steps whether there is $\G \ts p < q$,
  and if so, there is $\ \G \ts p <: q$.
\end{theorem}
\begin{proof}
  Follows from the decidability and soundness of $\Uparrow_1$ and $\Uparrow_2$,
  both terminating after visiting each $\G$ entry at most once.
  In $\Uparrow_1$, \rulename{qeself} produces a subqualifier of $q$ via \rulename{qself}.
  In $\Uparrow_2$, \rulename{qevar} produces a subqualifier of $q$ via \rulename{qvar}
  or \rulename{qtvar}.
  Together, they produce a subqualifier of $q$ containing $p$. 
  Thus, by \rulename{qsub} and \rulename{qtrans}, we have $\G \ts p <: q$.
\end{proof}

\begin{theorem}[Completeness of Qualifier Checking] \label{thm:qtp_complete}
  If $\ \G\ts p <: q$, and $\ \WF{\G}$, then $\ \G \ts p < q$.
\end{theorem}
\begin{proof}
  By induction on $\G \ts p <: q$.
  The key case \rulename{qtrans} shows that applying $\qexp[1]^*$
  and $\qexp[2]^*$ suffices, and further $\qexp[1]$ applications do
  not extend the result. This follows by induction on $\G$.
\end{proof}

Qualifier inference \rulename{qinfer} involves both qualifier exposure and
unification. Since exposure properties have been established in
qualifier checking, we establish results for qualifier unification.

\begin{lemma}[Decidability and Soundness of Unification]
  If $\ \G\ts q$ and $\ \WF{\G}$, then given $p$, it is decidable in finite
  steps whether there exists $\ \G' $ such that
  $\ \G \ts p \qunif q \ots{\G'}$.
  Moreover, if so, then $\ \G \sqsubseteq \G'$ and $\ \G' \ts p <: q$.
\end{lemma}
\begin{proof}
  The procedure terminates, visiting each entry in $\G$ at most once.

  Soundness follows by induction on $\G$.
  By the induction hypothesis and \Cref{thm:ctxgrow_on_hastype},
  \rulename{quself} is justified by \rulename{qself},
  and \rulename{quvar} by \rulename{qvar} or \rulename{qtvar}.
\end{proof}

\begin{theorem}[Decidability and Soundness of Qualifier Inference]
  If $\ \G\ts q$ and $\ \WF{\G}$, then it can be decided in finite steps whether there exists $\G'$
  such that $\ \G\ts p < q \dashv \algout{\G'}$, and if so, $\ \G \sqsubseteq \G'$, and $\ \G'\ts p <: q$.
\end{theorem}
\begin{proof}
  Follows by combining the decidability and soundness results for qualifier
  exposure and unification, established in the preceding lemmas.
\end{proof}

Since qualifiers form a preorder under subqualifying, qualifier inference cannot
guarantee completeness. 
As analyzed in \Cref{sec4:qualifier}, unification hides $p \setminus q'$ via self-references, computing set differences rather than true qualifier differences.
Consequently, our unification is incomplete.
Moreover, qualifier differences lack principal solutions, ruling out any complete unification algorithm.

\newtheorem{formalexample}[theorem]{Example}

\begin{formalexample}[No Principal Qualifier Difference]
  \label{thm:nonprincipal}
  Given a well-formed context $\G = a\!:\!(\TRef{\TBase}{})^{\qfresh},\,b\!:\!(\TRef{\TBase}{})^{\qfresh},\,c\!:\!\ty[a,b]{(\TRef{\TBase}{})}$, both $b$ and $c$ are valid qualifier differences
  between $c$ and $a$.
  That is, both $b$ and $c$ satisfy as a minimum $q$ such that $\G\ts c <: a, q$.
  Moreover, neither $\G\ts b <: c$ nor $\G\ts c <: b$ holds, so they are
  incomparable. This demonstrates the absence of a principal solution.
\end{formalexample}

Despite that qualifier difference is non-principal, $q'$ is a maximally exposed
qualifier that ensures completeness of \rulename{qcheck}, \ie, \Cref{thm:qtp_complete}.
Thus, $p \setminus q'$ still captures a minimal difference within the range of $p$,
even though other incomparable solutions may exist.

\subsection{Subtype Checking}

Subtype checking \rulename{sajoin} proceeds in two phases, and here we start by
discussing the first stage, self unpacking.

\begin{lemma}[Equivalence of Self Unpacking]\label{lemma:self-unpack-equiv}
  If $\ \G\ts \ty[q]{T}$ and $\ \WF{\G}$, $T'$ can be computed in finite steps
  such that $\ T \subt[1]^{q} \algout{T'}$.
  If so, then $\ \G\ts \ty[q]{T} \subt \ty[q]{T'}$ and $\G\ts \ty[q]{T'} \subt \ty[q]{T}$.
\end{lemma}
\begin{proof}
  By induction on $T$. The judgment $T \subt T'$ follows by the
  induction hypothesis and \rulename{qvar}, while the reverse direction
  $T' \subt T$ follows by the induction hypothesis and \rulename{qself}.
\end{proof}

As a preprocessing step,
\Cref{lemma:self-unpack-equiv} ensures soundness
and type equivalence between the unpacked type and the original,
inducing no additional incompleteness.
The decidability and soundness of subtype checking then reduce to
those of the recursive subtype checking phase.

\begin{theorem}[Decidability and Soundness of Subtype Checking]
  If $\ \G\ts \ty[q]{T}$, $\ \G\ts T'$, and $\ \WF{\G}$, then it can be decided in finite steps whether there exists such $q'$ and $\G'$ that $\ \G\ts \ty[q]{T} \subt[] \ty[q']{T'} \dashv \algout{\G'}$, and if so, $\ \G\ts q'$, $\ \qfresh\notin q'$, $\ \G \sqsubseteq \G'$, and $\G'\ts \ty[q]{T} \subt \ty[q']{T'}$.
\end{theorem}
\begin{proof}
  The result follows by applying the decidability and soundness of recursive
  subtype checking after invoking \Cref{lemma:self-unpack-equiv}.
  Soundness of combining the two steps relies on \rulename{strans}.

  Decidability is shown by measuring the size of $T$ and $T'$, with type variable size defined by its bound as a standard approach in kernel $F_{<:}$ decidability proofs (full definition omitted for brevity).

  Soundness is proven by induction on the checking function, using lemmas for
  manipulating qualifier holes, variants of \Cref{thm:ctxgrow_on_hastype,thm:gstighten_hastype}
  established for subqualifying and subtyping.
\end{proof}

As expected, subtype checking is incomplete due to the incompleteness of qualifier inference and the delta between full and kernel $F_{<:}$.
The declarative rules \rulename{srefl} and \rulename{strans} are also absent from \algolang.
While it's unclear if their omission contributes to incompleteness, we believe it does not.

\subsection{Avoidance.}

We show that the avoidance procedure preserves subtyping invariants.
We omit listing the lemmas for polarized substitution and refer readers to our artifact for details.

\begin{theorem}[Decidability and Soundness of Avoidance] \label{thm:avoid}
If $\ \G\ts T$, and $\ \G\ts q$, then
it can be decided in finite steps whether there are such $q', T'$ that
$\ \ty[q]{T} \ll_{z} \algout{\ty[q']{T'}}$ for a given variable $z$, and if so, there are
(1) $\G\ts \ty[q]{T} \subt \ty[q']{T'}$,
(2)~$q' \subs q,z$,
(3)~$z \notin T'$,
(4) $\forall x \notin T, x \notin T'$, and
(5) $\forall x \notin^- T, x \notin^- T'$.
\end{theorem}
\begin{proof}
  Follows by combining the decidability and soundness lemmas for both positive and negative substitution, which are proven by mutual recursive induction on type $T$.
\end{proof}

Our avoidance algorithm always selects the outermost self-reference.
This strategy is intuitive, and 
we strongly believe it produces minimal results without a formal proof, as illustrated below.
\begin{formalexample}[Minimality of Outermost Avoidance] \label{thm:minavoid}
  Given a well-formed context $\G = x\!:\!(\TRef{\TBase}{})^{\qfresh}$ and a well-formed type
  $T = f(\ty[\qbot]{B}) \to (g(\ty[\qbot]{B}) \to \ty[x]{(\TRef{\TBase}{})})^x$.
  To avoid $x$ inside, our avoidance algorithm uses $f$ and produces
  $T_1 = f(\ty[\qbot]{B}) \to (g(\ty[\qbot]{B}) \to \ty[f]{(\TRef{\TBase}{})})^f$.
  Alternatively, we can use $g$ to avoid the inner $x$, producing
  $T_2 = f(\ty[\qbot]{B}) \to (g(\ty[\qbot]{B}) \to \ty[g]{(\TRef{\TBase}{})})^f$,
  where $f$ is still necessary to avoid the second $x$.
  Moreover, there is $\ \G \ts T_1 <: T_2$, regardless of qualifiers.
\end{formalexample}

\subsection{Type Exposure and Bidirectional Typing.}
We establish the decidability and soundness of type exposure for
dealing with type variables in bounded quantification.

\begin{theorem}[Decidability and Soundness of Type Exposure]
  If $\ \G \ts T$, and $\ \WF{\G}$, then it can be decided in finite steps whether there exists such $T'$ that $\ \G \ts T \Uparrow \algout{T'}$, and if so, $\forall q, \G\ts \ty[q]{T} \subt \ty[q]{T'}$. 
\end{theorem}
\begin{proof}
  Termination follows by induction on $\G$, as well-formed contexts prevent reintroducing exposed type variables.
  Soundness follows by induction on the type exposure derivation.
\end{proof}

Finally, we show the decidability and soundness of bidirectional typing. 
Following results on qualifier inference and subtype checking,
our bidirectional typing \algolang is naturally incomplete to \maybelang.
Still, we will show that \algolang is expressive enough to check meaningful programs
in \Cref{sec:eval}.
\begin{theorem}[Decidability and Soundness of Bidirectional Typing]
  \label{thm:algtop}
  If $\ \WF{\G}$, then it can be decided in finite steps whether there exist such $\G'$, $\flt$, $T$, $q$ that $\G[\algout{\flt}] \ts t \Rightarrow \algout{\ty[q]{T}} \dashv \algout{\G'}$, and if so, $\ \G \sqsubseteq \G'$, $\ \G\ts \flt$, $\ \qfresh\notin\flt$, and $\ \cx[\flt]{\G'}\ts t : \ty[q]{T}$.
\end{theorem}
\begin{proof}
  Termination follows by induction on the term $t$, where transitions between modes are finite (from full checking  to checking, and then inference) for each $t$. 
  Soundness is proven by mutual induction across all modes, relying on the lemmas discussed in this section. 
\end{proof} 
\newpage
\bibliographystyle{ACM-Reference-Format}
\bibliography{references}

\end{document}